\documentclass[11pt]{article}

\usepackage{hyperref}
\usepackage{latexsym}
\usepackage{amssymb}
\usepackage{amsmath}
\usepackage{amsthm,fullpage}
\usepackage[noend]{algorithmic}
\usepackage{algorithm}
\usepackage{bm}

\usepackage{nth}

\usepackage{subfig}
\usepackage{graphicx}
\usepackage{enumitem}   
\usepackage{forloop}
\usepackage[paper=letterpaper, margin=1in]{geometry}
\usepackage{nicefrac}
\usepackage{rotating}
\usepackage{array,multirow}
\usepackage{threeparttable}
\usepackage{etoolbox}
\usepackage{tikz,tkz-graph}
\usetikzlibrary{arrows,shapes}
\usetikzlibrary{positioning,calc}
\usetikzlibrary{decorations.pathreplacing}
\usetikzlibrary{fadings}
\tikzset{>= stealth'}
\tikzstyle{vertex}=[circle, draw,fill=gray!20, inner sep=0pt, minimum size=16pt]
\tikzstyle{square}=[rectangle, draw,fill=gray!20, inner sep=0pt, minimum size=16pt]
\tikzstyle{svertex}=[circle, draw,fill=gray!20, inner sep=0pt, minimum size=12pt]
\tikzstyle{ssquare}=[rectangle, draw,fill=gray!20, inner sep=0pt, minimum size=12pt]
\tikzstyle{tvertex}=[circle, draw,fill=gray!20, inner sep=0pt, minimum size=8pt]
\tikzstyle{tsquare}=[rectangle, draw,fill=gray!20, inner sep=0pt, minimum size=6pt]

\newtheorem{theorem}{Theorem}[section]
\newtheorem{lemma}[theorem]{Lemma}

\newtheorem{corollary}[theorem]{Corollary}
\newtheorem{definition}[theorem]{Definition}

\newcommand{\OPT}{\operatorname{OPT}}
\newcommand{\Scal}{\mathcal{S}}

\newcommand{\xvec}{\pmb{x}}
\newcommand{\yvec}{\pmb{y}}

\newcommand{\tw}{\operatorname{tw}}
\newcommand{\link}{\operatorname{link}}
\newcommand{\paths}{\mathcal{P}}

\newcommand{\vol}{\operatorname{vol}}
\newcommand{\xvol}{\operatorname{vol}_x}
\newcommand{\yvol}{\operatorname{vol}_y}

\newcommand{\rmax}{1/12}
\newcommand{\markov}{48}
\newcommand{\sepLP}{\operatorname{sep-LP}}
\newcommand{\sepOPT}{\lambda}
\newcommand{\flowLP}{\operatorname{flow-LP}}
\newcommand{\gvec}{\pmb{g}}
\newcommand{\fvec}{\pmb{f}}
\newcommand{\feas}{\mathcal{F}}

\newcommand{\xval}{\operatorname{LP-cost}}
\newcommand{\yval}{\operatorname{wt}}
\newcommand{\Hhat}{\hat{H}}

\newcommand{\width}{\operatorname{width}}
\newcommand{\half}{\frac{1}{2}}

\newcommand{\dvec}{\pmb{d}}
\newcommand{\algpart}{\textsc{Partition}}
\newcommand{\Xconst}{\gamma}

\newcommand{\Hcal}{\mathcal{H}}
\newcommand{\algclassic}{\textsc{Decomposition}}

\newenvironment{proofof}[1][x]{\noindent{\em Proof of #1.}}{\qed\newline}

\DeclareMathOperator{\poly}{poly}

\newcommand\E{\mathbb{E}}

\newcommand{\defcal}[1]{\expandafter\newcommand\csname c#1\endcsname{{\mathcal{#1}}}}
\newcommand{\defbb}[1]{\expandafter\newcommand\csname b#1\endcsname{{\mathbb{#1}}}}
\newcounter{calBbCounter}
\forLoop{1}{26}{calBbCounter}{
    \edef\letter{\Alph{calBbCounter}}
    \expandafter\defcal\letter
		\expandafter\defbb\letter
}

\allowdisplaybreaks

\title{LP-Based Robust Algorithms for Noisy Minor-Free and Bounded Treewidth Graphs}

\author{Nikhil Bansal\thanks{Department of Mathematics and Computer Science, Eindhoven University of Technology, Netherlands.
Email:
\href{mailto:n.bansal@tue.nl}{n.bansal@tue.nl}.
Supported by NWO Vidi grant 639.022.211 and an ERC consolidator grant 617951.}
\and
Daniel Reichman\thanks{Electrical Engineering and Computer Science,
University of California at Berkeley,
Berkeley, CA 94709.
Email:
\href{mailto:daniel.reichman@gmail.com}{daniel.reichman@gmail.com}.}
\and
Seeun William Umboh\thanks{Department of Mathematics and Computer Science, Eindhoven University of Technology, Netherlands.
Email:
\href{mailto:seeun.umboh@gmail.com}{seeun.umboh@gmail.com}.
Supported by ERC consolidator grant 617951.}
}

\date{}

\begin{document}

\maketitle

\begin{abstract}
We give a general approach for solving optimization problems on noisy minor free graphs, where a $\delta$-fraction of edges and vertices are adversarially corrupted. The noisy setting was first considered by Magen and Moharrami and they gave a $(1 + \delta)$-estimation algorithm for the independent set problem. Later, Chan and Har-Peled designed a local search algorithm that finds a $(1 + O(\delta))$-approximate independent set. However, nothing was known regarding other problems in the noisy setting. Our main contribution is a general LP-based framework that yields a $(1 + O(\delta \log m \log \log m))$-approximation algorithm for noisy MAX-$k$-CSPs on $m$ clauses.

\end{abstract}

\pagenumbering{arabic}

\section{Introduction}
Several hard optimization problems often become substantially easier on special classes of graphs
such as planar graphs and bounded treewidth graphs. For example, while the maximum independent set problem is notoriously hard on general graphs \cite{hastad}, it admits an efficient approximation scheme on planar graphs \cite{LT,baker} and can be solved exactly in polynomial time on bounded treewidth graphs \cite{Bod}. Similarly, while MAX-$k$-SAT is APX-hard in general \cite{haastad2001some}, planar instances admit an efficient approximation scheme \cite{MotwaniKhanna} and bounded treewidth instances can be solved exactly in polynomial time \cite{freuder1990complexity,marx2007can}. 
In general, there has been extensive work done on designing better algorithms for special graph classes and several general techniques have been developed for this purpose. For problems on bounded treewidth graphs, several techniques based on dynamic programming and deep results from algorithmic graph minor theory and logic have been developed \cite{Fomin-book, Bod,DH08,Courcelle,eppstein2000diameter}. For problems on planar graphs, many surprising approximation guarantees
can be obtained based on decomposition approaches. One of the first decomposition approaches is based on the planar-separator theorem \cite{LT}. Later, Baker \cite{baker} developed a more versatile technique based on decomposition into $O(1)$-outerplanar graphs (which have bounded treewidth). For example, Khanna and Motwani \cite{MotwaniKhanna} used Baker's technique to obtain efficient approximation schemes for a wide variety of Constraint Satisfaction Problems (CSPs) such as MAX-SAT.

\medskip
\noindent {\bf Noisy graph models.} In this paper, we consider a natural question that was first studied by Magen and Moharrami \cite{MM}: What happens to these special graph classes when they are perturbed adversarially? For the maximum independent set problem, \cite{MM} considers the setting where an input graph $G$ on $n$ vertices is obtained from some (hidden) underlying planar graph $G_0$ by adding $\delta n$ arbitrary edges (these are called \emph{noisy} edges) for some small number $\delta >0$ and ask: how well can one approximate the maximum independent set (MIS) problem on $G$? More generally, one can consider the same question for other optimization problems that are easy on these special graph classes.

In this work, we consider MAX-$k$-SAT (for constant $k$) in the noisy setting. To introduce our noise model and to relate to the noisy graph model of \cite{MM}, we remind the reader of the definition of a factor graph: given a $k$-SAT formula $\Phi$, the \emph{factor graph} of $\Phi$ is a bipartite graph $H=(A,B)$ where $A$ contains a vertex for every variable appearing in $\Phi$, $B$ contains a vertex for every clause appearing in $\Phi$, and a clause-vertex $\phi$ is connected to a variable-vertex $x$ if and only if $x$ belongs to the clause $\phi$. In the noisy setting, the input formula $\Phi$ is given by an adversary who takes a planar $k$-SAT formula $\Phi_0$ (i.e.~$\Phi_0$'s factor graph is planar) with $n$ variables and $m$ clauses, and adds $\delta m$ clauses (each clauses contains exactly $k$ literals) to $\Phi_0$ (resulting in $\delta m$ vertices and $k \delta m = O(\delta m)$ edges being added to the factor graph of $\Phi_0$). We choose to focus on this noise model as it does not change the arity of the original formula $\Phi_0$. Our results carry over without much difficulty to more general noise models where one adds both vertices and edges to the factor graph of $\Phi_0$, so long as the total number of edges and vertices added is $\delta m$.

\medskip

\noindent {\bf Previous Work.}
Several previous works have considered approximation algorithms for MIS in the noisy setting. Magen and Moharrami presented an elegant argument showing that $\alpha(G)$ can be approximated to within a $(1+\epsilon)$ factor, for $\epsilon = \Omega(\delta)$, using $O(1/\epsilon)$ levels of the Sherali-Adams (SA) Hierarchy. Interestingly, this only yields an efficient {\em estimation} algorithm for $\alpha(G)$ and does not give any way to actually find the corresponding independent set.\footnote{Self-reducibility techniques such as those used for finding independent set in perfect graphs do not seem to work here as the estimation algorithm only gives an approximate answer.} In particular, the SA approach only uses the existence of $G_0$  to argue that $SA(G)\approx \alpha(G)$, without actually detecting the noisy edges. Later, Chan and Har-Peled \cite{ChanHP} developed a PTAS for planar independent set based on local search which can be used to obtain a $(1+\epsilon)$-approximation for MIS in noisy planar graphs. We are not aware of any previous work that studied approximation algorithms for noisy CSPs. 

\medskip
\noindent {\bf Limitations of Current Approaches.}
For MAX-$k$-SAT, classic approaches for planar instances based on computing decompositions of the factor graph break down completely in the presence of noise. In particular, the known algorithms for finding planar separators or decomposition into $O(1)$-outerplanar graphs inherently use the planar structure in various ways, and are easily tricked even with very few noisy edges.
In fact, a key motivation of \cite{MM} for studying the noisy setting was to design more ``robust" algorithms that are not specifically tailor-made for particular graph classes.
Note that in our noise model the adversary is even allowed to add a bounded
degree expander on some subset of $O(\delta n)$ vertices, completely destroying the planar structure and increasing the treewidth to $\Omega(n)$. 
Another natural approach might be to recover some planar graph $\tilde{G}$ from $G$ without removing too many edges, and then apply the algorithms for planar graphs to $\tilde{G}$.
However the best known guarantees for this Minimum Planarization problem are too weak for our purpose. In particular, even for bounded degree graphs these algorithms \cite{CMS,ChekuriS13} only achieve a $poly(n,\OPT)$
approximation, where $\OPT$ is the number of noisy edges, and thus only work in our setting when the noise parameter $\delta = n^{-\Omega(1)}$.

Finally, one may attempt to apply the local search algorithm of \cite{ChanHP}. However, the analysis of the local search algorithm of \cite{ChanHP} crucially relies on the existence of a decomposition based on the planar separator theorem known as $r$-divisions \cite{Frederickson87} and does not seem to apply to noisy instances of MAX-$k$-SAT.

\subsection{Our Contributions}
In this work, we give a general approach for solving problems on noisy planar graphs. Using this approach, we give the first algorithm that is able to handle noisy versions of planar MAX-$k$-SAT.

\begin{theorem}\label{thm:MAX-CUT}
  Let $\Phi_0$ be a planar $k$-SAT formula over $n$ variables and $m$ clauses, where $k$ is a constant independent of $n$ and $m$. Let $\Phi$ be a $k$-SAT formula obtained by adding $\delta m$ clauses to $\Phi_0$, for some $\delta > 0$. Then there is an algorithm such that given any $\epsilon > 0$, finds a $(1 + O(\epsilon + \delta \log m \log \log m))$-approximate assignment in time $m^{O(\log \log m)^2/\epsilon}$.
\end{theorem}

\noindent{\em Remark:}
Since MAX-$k$-SAT can be approximated within a constant factor, Theorem~\ref{thm:MAX-CUT} is of interest when $\delta=O(1/(\log m \log \log m))$. Theorem~\ref{thm:MAX-CUT} can also be extended to arbitrary (binary) $k$-CSPs.

For the maximum independent set problem, our approach gives an LP-based approximation algorithm with much better running time and whose approximation ratio has a better dependence on $\delta$. \begin{theorem}
  \label{thm:MIS}
  Let $G$ be an $n$-vertex graph obtained by adding $\delta n$ arbitrary edges to some planar graph $G_0$, for some $\delta>0$.
Then there is an algorithm such that given any $\epsilon > 0$, finds an independent set of size within a $(1 + O(\epsilon + \delta))$ factor of $\alpha(G)$, and runs in time $n^{O(1/\epsilon^4)}$.
\end{theorem}

\paragraph{General framework}

The starting point for our results is the following simple observation: Most algorithms for planar graphs use the planar structure only to find a certain structured decomposition of the graph, and once this is found, apply some simple or brute-force algorithm.

For example, the planar separator theorems are used to argue that given any $\epsilon >0$, a planar graph can be decomposed into disjoint components of size $O(1/\epsilon^2)$  by removing some subset $X$ of at most $\epsilon n$ vertices.
Similarly in Baker's decomposition, an $\epsilon$ fraction of edges $F$ (or vertices) can be removed from a planar graph to decompose it into $O(1/\epsilon)$-treewidth graphs.

Now consider a noisy planar graph $G$. We claim that such nice decompositions also exist for $G$ (although it is unclear how to find them). For example, consider the decomposition of the underlying planar graph $G_0$ into bounded size pieces
and for every noisy edge in $G\setminus G_0$ put one of its endpoint in $X$. This subset $X$ has size $|X| \leq (\epsilon + \delta)n$ and removing $X$ splits $G$ into components of size $O(1/\epsilon^2)$.
Similarly, let $F$ be the edges removed in Baker's decomposition of $G_0$, plus the noisy edges in $G\setminus G_0$. Clearly, $|F| \leq (\epsilon + \delta) n$ and removing $F$ decomposes $G$ into $O(1/\epsilon)$-treewidth graphs.

So this leads to the natural question of whether we can directly find such good decompositions,
without relying on the topological or other specific structure of the graph.
Our main contribution is to show that this can indeed be done using general LP-based techniques.
Specifically, we consider the following problems.

\begin{enumerate}
\item {\em Bounded Size Interdiction (weaker form):} Suppose $G$ can be decomposed into components of size at most $1/\epsilon^2$ by removing some (small) subset $X$ of vertices. Find such a subset $X$.
\item {\em Bounded Treewidth Interdiction:} Suppose $G$ can be decomposed into graphs of treewidth at most $w$, by removing some subset $F$ of edges. Find such a subset $F$.
\end{enumerate}

\paragraph{Our Results.}
We show the following results for these general interdiction problems. Our main technical result is the following theorem.

\begin{theorem}
\label{thm:bti} Given a graph $G$ and an integer $w>0$,
let $F$ be some subset of edges such that removing them reduces the treewidth of $G$ to $w$.
Then there is an algorithm that runs in time $n^{O(1)}$ and finds a subset of edges $F'$ such that $|F'| = O(\log n \log \log n) |F|$ and removing $F'$ from $G$ reduces the treewidth to $O(w \log w)$.
\end{theorem}\label{thm:treewidth_intrediction}
Note that this gives a bicriteria $(\log n \log \log n, \log w)$-approximation to the Bounded Treewidth Interdiction problem. We also remark that the approximation factors and the running time of the above algorithm do not depend on $w$.

Let us now see how Theorem~\ref{thm:bti} can be used to obtain Theorem~\ref{thm:MAX-CUT}. The basic idea is to apply Theorem~\ref{thm:bti} to the factor graph $H$ of the noisy formula $\Phi$ to obtain a formula $\Phi'$ whose factor graph has a smaller treewidth, and then compute an exact solution for $\Phi'$. Since $k$ is constant, the noisy formula $\Phi$ has at most $O(m)$ variables and thus $H$ has $\Theta(m)$ vertices. Thus, as discussed above, we know that there exists a set $F$ of $(k\delta + O(\epsilon /(\log m\log \log m))) m$ edges that we can delete from $H$ to reduce its treewidth to $O(\log m \log \log m/\epsilon)$. So we can apply Theorem~\ref{thm:bti} to $H$ with $w = O(\log m \log \log m / \epsilon)$ to find an edge set $F'$ of size at most $O((\epsilon + k\delta (\log m \log \log m))m)$ such that the residual factor graph $H - F'$ has treewidth $O(\log m (\log \log m)^2/\epsilon)$. Now observe that deleting the clauses incident to $F'$ gives us a formula $\Phi'$ whose factor graph is a subgraph of $H-F'$, and thus has treewidth $O(\log m (\log \log m)^2/\epsilon)$. Thereafter, one can use an exact algorithm for MAX-$k$-SAT \cite{MotwaniKhanna} that runs in time $2^{O(\ell)}\cdot \poly(m)$ when the factor graph has treewidth at most $\ell$. 
Next, we analyze the quality of this solution. Denote by $\OPT$ and $\OPT'$ the maximum number of satisfiable clauses in $\Phi$ and $\Phi'$, respectively. At least a constant fraction of the clauses in $\Phi$ are satisfiable, so $\OPT \geq \Omega(m)$.
Since the formula $\Phi'$ is obtained by deleting $O(\epsilon + \delta (\log m \log \log m))m$ clauses from $\Phi$ (recall that $k$ is a constant), we have $\OPT' \geq \OPT - O(\epsilon + \delta (\log m \log \log m))m$. Thus, $\OPT' \geq (1 - O(\epsilon + \delta \log m \log \log m))\OPT$ and so our solution is a $(1 + O(\epsilon + \delta \log m \log \log m))$-approximation. 

We remark that Bounded Treewidth Interdiction is a fundamental problem of interest beyond its application to noisy planar graphs. For instance, treewidth interdiction algorithms have found applications in designing PTASes in several settings \cite{FLMS12,DHK,fomin2011bidimensionality}.

Next, we consider the Bounded Size Interdiction problem in Section \ref{sec:MIS}.

\begin{theorem} (Weaker version)
\label{thm:weak-bsi}
For the weaker form of the Bounded Size Interdiction problem stated above, given any $\beta \leq  1$, we can find in time $n^{O(1/\epsilon^2)}$ a subset of vertices $X'$
with $|X'| \leq O(|X|/\beta + \beta |E|)$ such that $G[V\setminus X']$ has no component larger than $1/\epsilon^2$. Here, $|E|$ is the number of edges in $G$. \end{theorem}

For a noisy planar graph $G$, there exists $X \subseteq V$ of size $|X| \leq (\epsilon + \delta)n$ (as discussed above) and $|E| \leq (3+\delta)n = O(1)n$. The latter follows as a planar graph has at most $3n-6$ edges. So setting $\beta = (\epsilon+\delta)^{1/2}$ in Theorem \ref{thm:weak-bsi} already gives $X'$ of size $O((\epsilon + \delta)^{1/2} n)$ in time $n^{O(1/\epsilon^2)}$. This can be used to design a $(1 + O((\epsilon + \delta)^{1/2}))$-approximation algorithm for independent set in noisy planar graphs. To get the better approximation factor of $1 + O(\epsilon + \delta)$ of Theorem \ref{thm:MIS} above, we will use a more refined result (Theorem \ref{lem:noisy-LT}) that decouples the dependence of $|X'|$ on $\epsilon$ and $\delta$.

To prove Theorem \ref{lem:noisy-LT}, we write a configuration LP based formulation and round it suitably. (See Section \ref{sec:MIS}.) The proof of Theorem \ref{thm:bti} is much more challenging and requires several new ideas, and we give a broad overview of the algorithm and the proof below.

\paragraph{Overview of Techniques for Theorem \ref{thm:bti}:}
First, observe that if $G$ has treewidth at most $w$, then $F=\emptyset$, and the algorithm must return $F'=\emptyset$. Thus, the problem is at least as hard as determining the treewidth of $G$.
This is well known to be NP-Hard, and in fact unlikely to admit a polynomial time $O(1)$ approximation under reasonable complexity assumptions \cite{WAPL}.
This implies that the bicriteria guarantee is necessary, and that it is unlikely  that the approximation with respect to $w$ can be made $O(1)$.

At a high level, our algorithm will try to construct a good tree decomposition of width $w$, while removing some problematic  edges along the way. To this end, let us first see how the known algorithms for finding tree decompositions work.
Treewidth is characterized up to $O(1)$ factor by the well-linkedness property of a graph, which allows one to construct a tree decomposition by computing small balanced vertex separators recursively. Either the algorithm succeeds at each step and eventually finds a tree decomposition, or returns a well-linked set as a certificate that $G$ has large treewidth. Finding balanced vertex separators is hard but one can use LP or SDP (resp.) formulations \cite{BGHK95,Amir,FHL} based on spreading constraints \cite{ENRS,ARV}, and lose an $O(\log w)$ or $O(\sqrt{\log w})$ (resp.) factor in the quality of the treewidth.

In the noisy setting, our algorithmic task can thus be viewed as detecting which edges to remove so that the above recursive procedure works. To do this, we formulate an LP with variables for which edges to remove (let us call these the $x_{uv}$ variables) so that in the residual graph every subset $S$ of vertices has a small fractional balanced vertex separator of size at most $w$.
However, as there are exponentially many such sets $S$, this gives a huge overall LP with (both) exponentially many variables and constraints, and it is unclear how to solve it. In particular, we have exponentially many different vertex separator LPs coupled together with the common $x_{uv}$ variables.

We describe the algorithm in two parts.
First, we assume that we are given the $x_{uv}$ values from some feasible optimum LP solution.
Using these $x$-values, for any given set $S$, we can now formulate a balanced edge-and-vertex separator LP, where the $x$-values give the fractional amount by which edges are removed, and in addition at most $w$ vertices are removed. Using standard region-growing techniques jointly on these edge and vertex values, we decide which edges to delete (this adds to $F'$), and which vertices lie in the separator for $S$ (these enter the bags in the tree decomposition).
Doing this directly gives an $O(\log^2 n)$ approximation (provided we ensure that the separator tree is balanced and has depth $O(\log n)$), due to the loss of an $O(\log n)$ factor on each level of recursion. To reduce this to $O(\log n \log \log n)$, we use ``Seymour's trick" of more careful region-growing \cite{Seymour}, together with some additional technical steps needed to make it work together with the tree decomposition procedure.

Second, we describe how to ``solve" the LP. Perhaps surprisingly, this turns out to be quite challenging and requires some new ideas, which may be be useful in other contexts. We only sketch these here, and details can be found in Section \ref{sec:LP}.
First, we bypass the need to completely solve this LP, by using the
Round-or-Separate framework (as in \cite{AnSS14,LLS}). In particular,
the algorithm starts with some possibly infeasible solution $x$, and tries to construct the tree decomposition.
If it succeeds, we are done. Otherwise, it gets stuck at finding a small balanced vertex separator for some set $S$. At that point, we try to add a violated inequality.
However, a crucial point is that we need to find a violated inequality only involving the $x$-variables. So, a key step is to reformulate the LP to only have the $x$-variables.
This crucially uses LP duality and the structure of the LP that the variables for different sets $S$ are only loosely coupled via the $x$-variables. After this reformulation, it is still unclear how to find
a violated inequality due to the exponential size of the LPs involved. We get around this issue by using some further properties of the Ellipsoid Method and the LP duality.

\paragraph{Other Related Work.} The noise model considered here gives an interesting interpolation between easy and general worst-case instances. This is similar in spirit to approaches such as smoothed analysis \cite{SpielmanTeng}, planted models and semi-random models \cite{FK98,MMV12}, although unlike these models our noise model is completely adversarial.
For the general version of the bounded size interdiction problem---given a graph $G$ and a size parameter $s$, find the smallest vertex set $X$ to delete so that $G - X$ has components of size at most $s$---Lee~\cite{lee17small} independently gave a $O(\log s)$-approximation algorithm that runs in time $2^{O(s)}\poly(n)$. However, using this for maximum independent set on noisy planar graphs only yields a $(1 + O((\delta +\epsilon)\log (1/\epsilon))$-approximation. Fomin et al.~\cite{FLMS12} have considered the vertex deletion variant of our treewidth interdiction problem. They obtain a constant approximation factor for the problem, however their approximation factor depends at least exponentially on $w$ which makes it inapplicable in our setting\footnote{We need a sublinear dependence on $w$. To see why, consider for example the noisy MIS problem with $G_0$ as a grid. If we wish to reduce the treewidth to $w$, we would need to remove an $\Omega(1/w)$ fraction of the vertices, so if the interdiction algorithm is not an $o(w)$ approximation, it might end up deleting all the vertices.}. Also their algorithm is polynomial time only for $w=O(1)$. A related model was considered by \cite{HKNO} in the context of property-testing and sublinear time algorithms in the bounded degree model.

\section{Notation and Preliminaries}
We always use $G_0$ for the underlying graph, and $G$ for the noisy graph. The number of vertices of $G$ is always $n$.
For a subset $S \subseteq V$ and $F \subseteq E$, we use the notation $G[S] - F$ to denote the subgraph induced on the vertices $S$, excluding the edges in $F$. We use the notation $E(S)$ to denote the subset of $E$ with both endpoints in $S$, and given another subset $S' \subseteq V$, we use $E(S,S')$ to denote the subset of $E$ with one endpoint in $S$ and another in $S'$. 

\paragraph{Planar and Minor-Free graphs.}
The classic planar separator theorem \cite{LT} states that any planar graph has a $2/3$-balanced vertex separator of size $O(\sqrt{n})$ and that it can be found efficiently. Applying this recursively gives the following.
\begin{lemma}
\label{thm:LT}
For any planar graph $G$ and any $\alpha>0$, there is subset of vertices $X \subset V$ with $|X| =O(\alpha n)$, such that every component $C_i$ of $G_0[V-X]$ has at most $1/\alpha^2$ vertices. \end{lemma}
A more generally applicable technique (see e.g.~\cite{MotwaniKhanna,eppstein,Demaine2008}) is Baker's decomposition \cite{baker}, which states that for any integer $k$, a planar graph can be decomposed into pieces of treewidth $O(k)$ (specifically, $k$-outerplanar graphs) by removing $O(1/k)$ fraction of edges or vertices. 

A minor of $G$ is a graph $G'$  obtained by deleting and contracting edges. A graph $G$ is $H$-free
if $G$ does not contain a subgraph $H$ as a minor. Planar graphs
are exactly the graphs excluding $K_{3,3}$ and $K_5$ as minors. In fact, Robertson and Seymour proved that every graph family
closed under taking minors is characterized by a set of
excluded minors.
Both the planar separator theorem and Baker's decomposition approach extend more generally to 
$H$-free graphs \cite{AST,devos,DHK}. 

\paragraph{Treewidth.} We review some relevant definitions related to treewidth.

\begin{definition}[$\alpha$-separator of $S$ in $G$]
Given a graph $G=(V,E)$ and a set $S \subset V$,
a vertex set $X \subset V$ is an \emph{$\alpha$-separator} of vertex set $S$ in $G$ if every component $C$ of $G[V - X]$ has $|C \cap S| \leq \alpha|S|$. 
\end{definition}

\begin{definition}[Well-linked sets]
  A vertex set $S$ is \emph{$w$-linked} in $G$ if it does not have a $\half$-separator $X$ with $|X| < w$. The \emph{linkedness} of $G$ is defined to be the maximum integer $w$ such that there exists a $w$-linked set in $G$, and is denoted as $\link(G)$.
\end{definition}

\begin{definition}[Tree decomposition]
  \label{def:treedec}
  A \emph{tree decomposition} of $G$ is a tree $T$ whose nodes $t$ correspond to vertex subsets $V_t$ of $G$ (called bags) that satisfies the following properties:
       (i) for every edge $(u,v) \in E$, there exists a bag $V_t$ containing both $u$ and $v$;
  (ii) for every vertex $v$, the bags that contain $v$ form a non-empty subtree of $T$.     The \emph{width} of the decomposition is $\width(T) = \max_{s \in T} |V_s| - 1$.
\end{definition}

\begin{definition}[Treewidth]
  The \emph{treewidth} of $G$ is the minimum integer $w$ such that it has a tree decomposition of width $w$, and is denoted as $\tw(G)$.
\end{definition}

The following well-known (see e.g. \cite{Reed97}) approximate characterization of treewidth in terms of linkedness will be useful for us.

\begin{lemma}[\cite{Reed97}]
  \label{lem:tw-duality}
	For any graph $G$, 
  $\link(G) < \tw(G) < 4\link(G)$.
\end{lemma}

\newcommand{\interdict}{\textsc{Interdict}}
\newcommand{\near}{\frac{2}{3}}
\newcommand{\xcut}{\delta_x}
\newcommand{\ycut}{\delta_y}
\newcommand{\const}{\kappa}
\section{Bounded Treewidth Interdiction}
Recall that in the Bounded Treewidth Interdiction problem, we are given a graph $G = (V,E)$, a target treewidth $w$, and we want to find the minimum set $F$ of edges to delete from $G$ such that $\tw(G-F) < w$. In this section, we describe the exponential-size LP and sketch the rounding algorithm used to prove Theorem \ref{thm:bti}. In the following, when $X$ is a vertex set and $F$ is an edge set, we use the shorthand $G-X$ to mean $G[V - X]$, and $G-X-F$ to mean $G[V - X] - F$.

\subsection{An Exponential-Sized LP}
\label{sec:lp-relaxation}
Lemma \ref{lem:tw-duality} gives us a convenient characterization of feasible solutions $F$ which we can use to write an LP. In particular, it says that if $\tw(G-F) < w$, then every vertex set $S \subseteq V$ has a $\half$-separator $X^S$ in $G - F$ of size less than $w$. Consider LP \eqref{lp:exp} in Figure \ref{fig:lpexp}. It has a variable $x_{uv}$ indicating if edge $(u,v) \in E$ belongs to $F$. For every subset $S\subseteq V$ and vertex $v \in V$, variable $y^S_v$ indicates if $v$ belongs to the minimum-size $\half$-separator $X^S$ of $S$ in $G-F$. For $u,v \in V$, let $\paths(u,v)$ denote the set of paths between $u$ and $v$. For a path $P$, define $E(P)$ to be the set of edges in $P$ and $V(P)$ to be the set of vertices on $P$, including the endpoints.

\label{sec:ROS}
\begin{figure}
\centering
\begin{equation}
    \boxed{
      \label{lp:exp}
      \begin{aligned}
        \mbox{min} \quad
        & \sum_{(u,v) \in E} x_{uv}\\
        \mbox{s.t.}\quad
        & \sum_{v \in V} y^S_v \leq w &\quad \forall S \subseteq V\\
        & d^S_{uv} \leq \sum_{e \in E(P)} x_e + \sum_{t \in V(P)} y^S_t &\quad \forall S \subseteq V, u,v \in V, P \in \paths(u,v) \\
        & \sum_{v \in U} d^S_{uv} \geq |U| - \frac{|S|}{2} &\quad \forall U \subseteq S \subseteq V, u \in U
      \end{aligned}
    }
  \end{equation}
  \caption{LP relaxation for the treewidth interdiction problem.}
  \label{fig:lpexp}
\end{figure}

We interpret the solution as follows: the LP assigns a length $x_{uv}$ to each edge $(u,v) \in E$ and a weight $y^S_v$ for each vertex set $S$ and vertex $v$. Consider a fixed set $S$. Without loss of generality, we can assume that the variable $d^S_{uv}$ denotes the distance between $u$ and $v$ induced by the edge lengths $x_e$ and vertex weights $y^S_t$. In particular, if we define the length of a path $P \in \paths(u,v)$ to be the sum of edge lengths and vertex weights on the path, including the weights on $u$ and $v$, then $d^S_{uv}$ is the length of the shortest path between $u$ and $v$. The variables $d^S_{uv}$ and the last set of constraints are often called \emph{spreading metrics} and \emph{spreading constraints}, respectively.
It is also easy to see
that without loss of generality any feasible solution satisfies $d_{uv}^S \leq 1$. Note that there is a potentially different metric $d^S$ for each set $S$, and that the LP has
exponentially many constraints and exponentially many variables.

\begin{lemma}
  \label{lem:LP-relax}
  LP \eqref{lp:exp} is a relaxation of the treewidth interdiction problem.
\end{lemma}

\begin{proof}
  We show that for every edge set $F$ such that $\tw(G-F) < w$, there exists a feasible solution $(\xvec,\yvec,\dvec)$ to LP \eqref{lp:exp} with $\sum_{(u,v) \in E} x_{uv} \leq |F|$. Let $\xvec$ be the indicator vector for $F$.

As $\tw(G-F)< w$, by Lemma \ref{lem:tw-duality}, for each vertex set $S$, there exists a set $X^S$ of at most $w$ vertices such that no component of $G-X^S-F$ contains more than half of $S$. Define $\yvec^S$ to be the indicator vector for $X^S$ and $d^S_{uv}=1$ if either $u$ or $v$ lies in $X^S$, or if $u$ and $v$ lie in separate components of $G-X^S-F$.

The solution $(\xvec,\yvec,\dvec)$ has $\sum_{(u,v) \in E} x_{uv} = |F|$ and satisfies the first two sets of constraints. It remains to show that it satisfies the spreading constraints. Fix some $S$ and consider $U \subseteq S$ with $|U| > |S|/2$. Let $u$ be a vertex of $U$. There are two cases to consider:
(i) If $u \in X^S$, we have $d^S_{uv} = 1$ for all $v \in U \setminus \{u\}$ and so $\sum_{v \in U}d^S_{uv} \geq |U| - 1 \geq |U| - |S|/2$;
(ii) otherwise if $u \notin X^S$, let $C$ be the component of $G-X^S-F$ that contains $u$. We have $|C \cap S| \leq |S|/2$ since $X^S$ is a $\half$-separator of $S$ in $G-F$. We also have $d^S_{uv} = 1$ for every $v \in U - (C \cap S)$ since $v$ is either in a different component of $G-X^S-F$ or in $X^S$. Thus,
\[\sum_{v \in U} d^S_{uv} \geq \sum_{v \in U - (C \cap S)} d^S_{uv} \geq |U| - |C \cap S| \geq |U| - |S|/2.\]
Therefore, the solution $(\xvec,\yvec,\dvec)$ is a feasible solution with $\sum_{(u,v) \in E} x_{uv} = |F|$.
\end{proof}

In the rest of this section, we describe our rounding algorithm for LP \eqref{lp:exp} and prove the following lemma.
\begin{lemma}
  \label{lem:tw-oracle}
  Given oracle access to a feasible solution $(\xvec, \yvec, \dvec)$ of LP \eqref{lp:exp}, we can find in time $\poly(n)$ an edge set $F$ such that $\tw(G-F) \leq O(w \log w)$ and $|F| \leq O(\log n \log \log n)\sum_{(u,v) \in E}x_{uv}$. \end{lemma}

\subsection{Sketch of  the Rounding Algorithm}
\label{sec:rounding-sketch}
Let $(\xvec,\yvec,\dvec)$ be a feasible solution to LP \eqref{lp:exp}. As mentioned in the Introduction, our rounding algorithm is based on a recursive algorithm for constructing tree decompositions with Seymour's recursive graph decomposition trick layered on top. We now describe at a high level how these ideas are combined together and highlight the key issues that arise.

\paragraph{Classic Tree Decomposition.} We begin by outlining the relevant parts of the classic tree decomposition algorithm \cite{RobertsonS95b} (which we call $\algclassic$). It is a recursive algorithm: Given a subgraph $H$ of $G$ and an integer $w$, it either constructs a tree decomposition of width $O(w)$ or it finds a $w$-linked set $S$ which certifies\footnote{It certifies $\tw(H) \geq w$, and the treewidth of a graph is at least the treewidth of any subgraph.} that $\tw(G) \geq w$. There are two key steps (illustrated in Figure \ref{fig:recursive-separator}) that are important to us:
\begin{itemize}
\item Separate: find a minimum-size $\near$-separator $X$ of a vertex set $S$ in $H$ (the particular choice of $S$ depends on previous recursive steps). \item Recurse: for each component $C_i$ of $H-X$, recurse on the subgraph $H_i$ of $H$ which consists of the edges of $H$ induced by $C_i$ and those between $C_i$ and $X$.
\end{itemize}
We say that $\algclassic$ ``succeeds'' if it does not encounter a $w$-linked vertex set $S$. In particular, when $\algclassic$ succeeds, the separators found during its execution can be used to construct a tree decomposition of width $O(w)$; when it fails, it has found a $w$-linked set $S$ during its recursion. 

\begin{figure}
  \centering
  \hfill
  \includegraphics[scale=0.3]{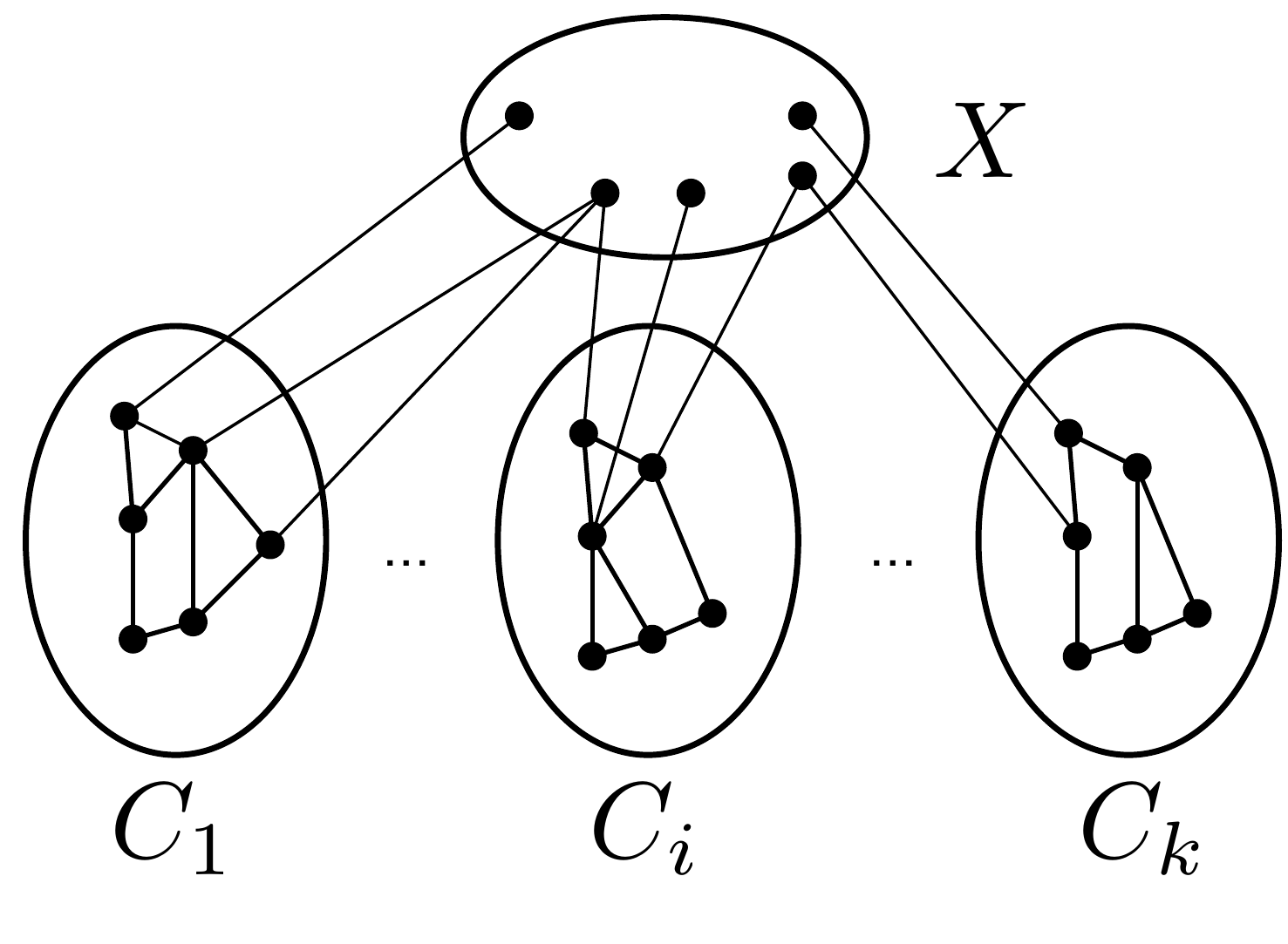}
  \hfill{}
  \caption{$\algclassic$ finds a separator $X$, and for each component $C_i$ of $H-X$, it recurses on the subgraph $H_i$ consisting of edges of $H$ with at least one endpoint in $C_i$.}
  \label{fig:recursive-separator}
\end{figure}

\paragraph{Our Algorithm.}
Our algorithm (which we call $\interdict$) largely follows along the lines of $\algclassic$. The main difference is that we also want to delete edges to ensure that size of the separators found in the recursion are small enough so that we succeed in constructing a tree decomposition of width $O(w\log w)$. We still make the same choices about which set $S$ to separate and how to recurse; this allows us to reuse the analysis of $\algclassic$ to prove that we have deleted enough edges to reduce the treewidth down to $O(w \log w)$. In particular, instead of the Separate step, we want to perform a ``Delete and Separate'' step instead.
\begin{itemize}
\item Delete and Separate: delete a subset $D$ of edges and find a vertex set $X$ of size $O(w \log w)$ such that $X$ is a $\near$-separator of $S$ in $H-D$.
\end{itemize}
Here is where LP \eqref{lp:exp} is useful. It is similar to the spreading metric relaxation for finding minimum balanced vertex separators, except that it gives edge-and-vertex separators. As mentioned in the Introduction, one can apply standard region growing techniques in an almost black box fashion along with other tricks to obtain a $(\log^2 n, \log w)$-approximation to treewidth interdiction.

\begin{figure*}
  \centering
                \hfill
  \subfloat{\label{fig:embedding}\includegraphics[scale=0.3]{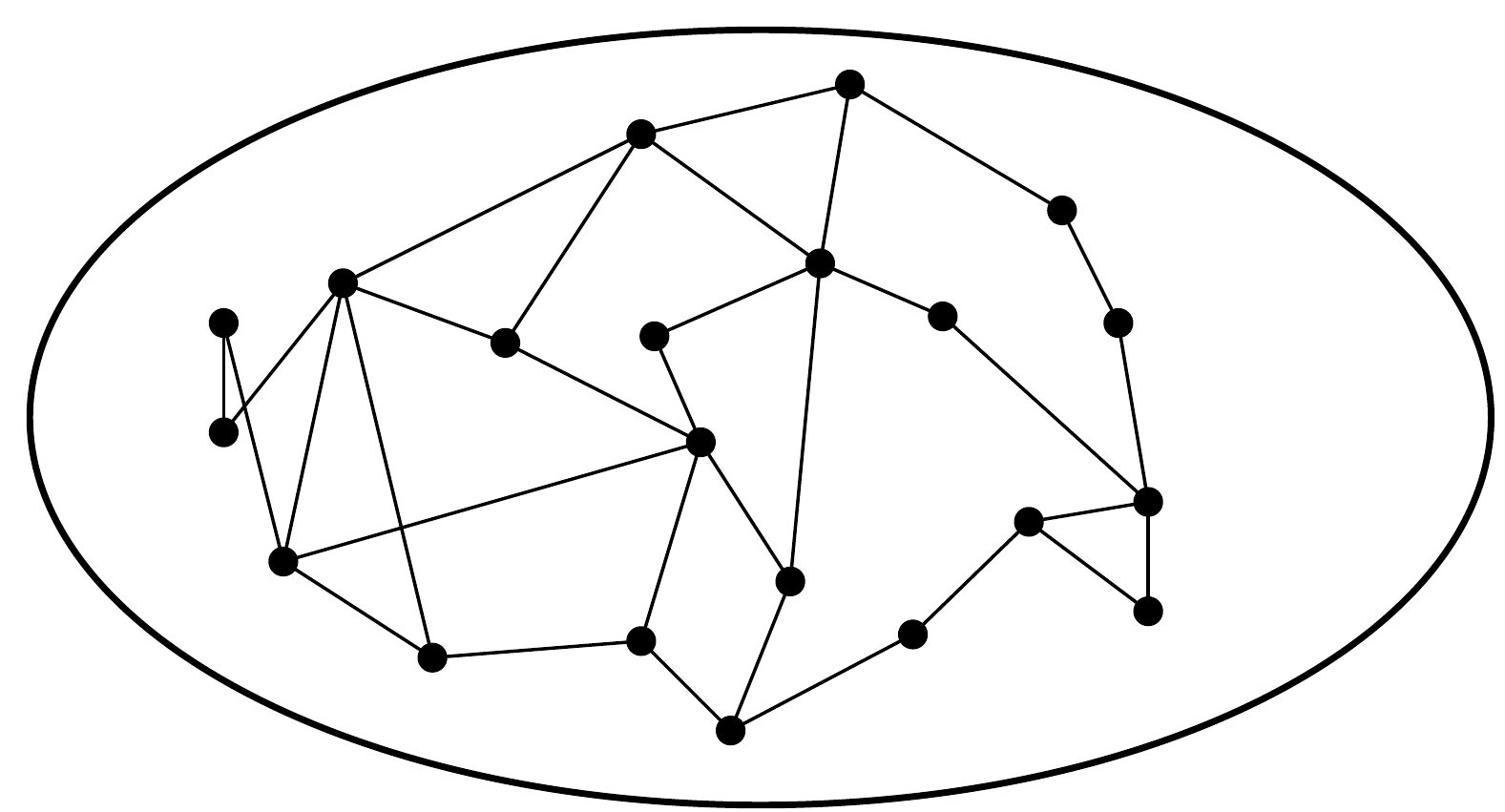}}
  \hfill
  \subfloat{\label{fig:regiongrowing}\includegraphics[scale=0.3]{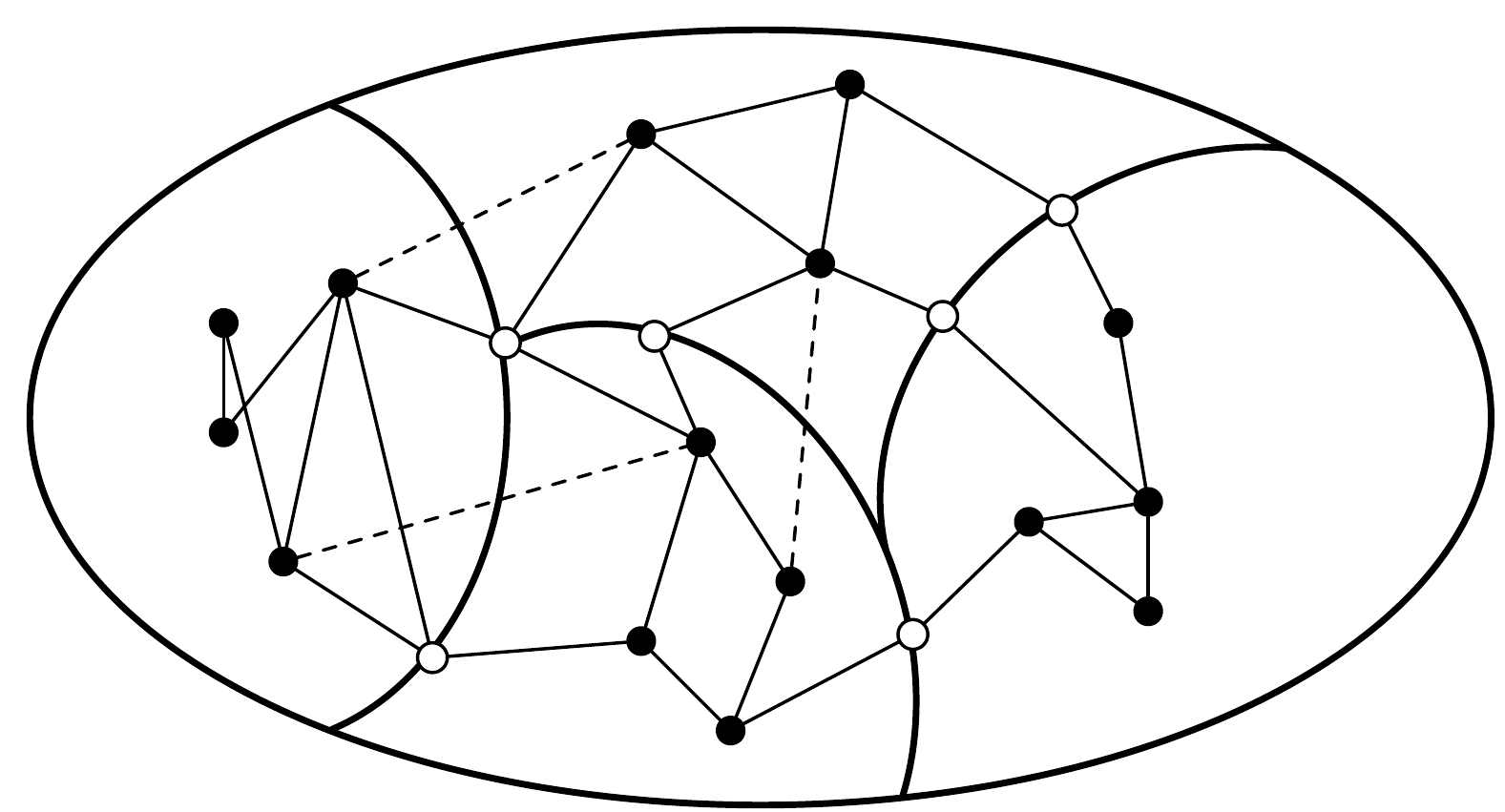}}
  \hfill
  \subfloat{\label{fig:region-subgraph}\includegraphics[scale=0.3]{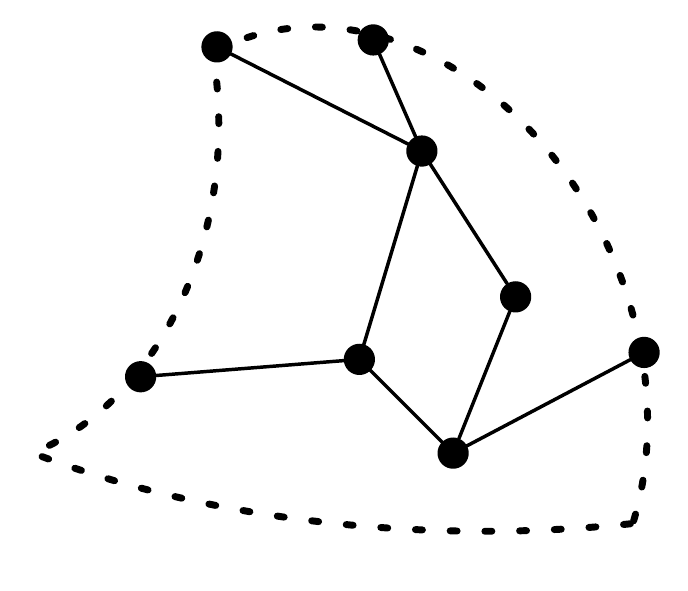}}
  \hfill{}
  \caption{$\interdict$ uses the spreading metric given by LP solution (left), partitions the metric into regions with cut vertices shown as hollow vertices and cut edges shown as dashed edges (center), and for each region, recurses on the subgraph contained in it (right). Note that the subgraph includes the cut vertices on the boundary of the region as well.}
  \label{fig:partitioning}
\end{figure*}

Obtaining a $(\log n \log \log n, \log w)$-approximation needs some care. At a high level, we want to apply region growing recursively using Seymour's recursion. The basic idea is to ensure that not only is $|D|$ bounded by the cost of the LP solution projected on $H$, but it is also bounded in some nice way by the cost of the LP solution on the subgraphs $H_i$ we recurse on. In particular, we want to use region growing to partition the spreading metric $d^S$ into ``regions'' $B_1, \ldots, B_k$ (illustrated in Figure \ref{fig:partitioning}) with the following properties:
\begin{itemize}
\item (bounded charge) the set of edges $\xcut(B_i)$ (vertices $\ycut(B_i)$ resp.) cut by $B_i$ can be charged to the $x_{uv}$ variables ($y^S_v$ variables, resp.) ``contained'' in the region,
\item (containment) each subgraph $H_i$ we recurse on is contained inside a region.
\end{itemize}
We can then choose $D = \xcut(B_1) \cup \cdots \cup \xcut(B_k)$ and $X=\ycut(B_1) \cup \cdots \cup \ycut(B_k)$. Due to the structure of the subproblems in the Recurse step, ensuring the second property requires some care with how we implement region growing.
The problem is that the separators involve both edges and vertices, and they play different roles, i.e.~edges are deleted globally, while vertices are deleted locally for each $S$.
Normally, the region growing technique proceeds by finding a region $B_i$ in the current graph that satisfies the bounded charge property, removes $B_i \cup \ycut(B_i)$---the vertices contained in or cut by $B_i$ from the graph---and repeats on the residual graph. However, this would also remove any edge $(u,v)$ with $u \in \ycut(B_i)$, even though $v$  is still remaining in the residual graph. In this case, no future region can contain or cut the edge $(u,v)$. This is a problem if $v$ ends up being in some component $C_i$ later, as we would have a subgraph $H_i$ that is not contained in any region.Thus, we need to somehow preserve edges between $\ycut(B_i)$ and the residual vertices. We will do this by introducing copies of these edges that we need to preserve, called ``zombie edges". A detailed description of Delete and Separate step as well as the rounding algorithm appears in the next subsection.


\subsection{The Delete and Separate Step}
\label{sec:algpart}

We start by describing the Delete and Separate step sketched out in Section \ref{sec:rounding-sketch}. The main ingredient is the region growing technique.

\paragraph{Region Growing.}
\label{sec:region}
We begin with some standard definitions. We define the ball $B(s,r)$ centered at vertex $s$ with radius $r$ as $B(s,r) = \{v : d^S_{sv} \leq r\}$. We say that a vertex $v$ is \emph{cut} by $B(s,r)$ if $d^S_{sv} - y^S_v < r < d^S_{sv}$, i.e.~$v$ is only ``partially inside'' the ball, and use $\ycut(s,r)$ to denote the set of vertices cut by $B(s,r)$; likewise, an edge $(u,v)$ is \emph{cut} by $B(s,r)$ if $d^S_{su} \leq r < d^S_{su} + x_{uv}$ and we use $\xcut(s,r)$ to denote the set of edges cut by $B(s,r)$.
We define two quantities for the $x$-cost and $y$-cost.
The \emph{cost} $\xval(s,r)$ of $B(s,r)$ is defined to be the total cost of the LP solution ``contained'' in $B(s,r)$:
 \begin{align*}
   \xval(s,r) = \sum_{u \in B(s,r), v \in B(s,r) \cup \ycut(s,r)} x_{uv}\\
    +\sum_{(u,v) \in \xcut(s,r) : u \in B(s,r), v \notin B(s,r)} r - d^S_{su}
  \end{align*}

Similarly, its \emph{weight} is defined to be the total fractional weight that is ``contained'' in $B(s,r)$:
\[\yval(s,r) = \sum_{v \in B(s,r)} y^S_v + \sum_{v \in \ycut(s,r)}   r - (d^S_{sv} -y^S_v).\]
See Figure \ref{fig:region} for an illustration.
The \emph{$x$-volume} and \emph{$y$-volume} is then defined to be $\xvol(s,r) =  \frac{\xval(G)}{n^2} + \xval(s,r)$ and $\yvol(s,r) = \frac{1}{w} + \yval(s,r)$.

These notions also extend to subgraphs. Given a subgraph $H'$, we define $\xval(H') = \sum_{(u,v) \in H'} x_{uv}$ and $\yval(H') = \sum_{v \in H'} y^S_v$. Similarly, $\xvol(H') =  \frac{\xval(G)}{n^2} + \xval(H')$ and $\yvol(H') = \frac{1}{w} + \yval(H')$.

\begin{figure}
  \centering
  \includegraphics[scale=0.4]{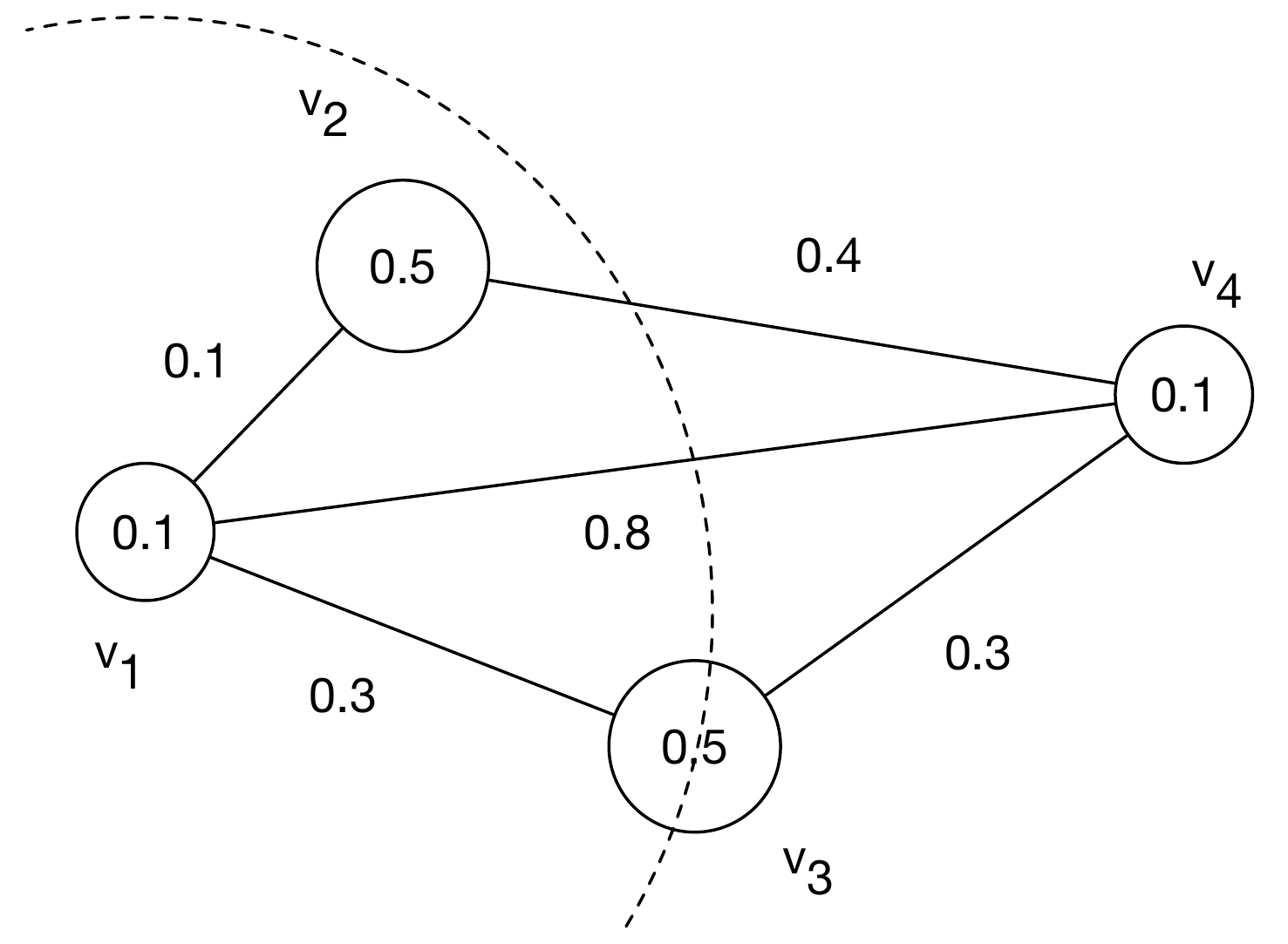}
  \caption{In this figure, the edge lengths represents the $x$ variables, and vertices are represented as circles whose diameters correspond to their $y^S$ weights. The dashed arc is centered at $v_1$ and has radius $0.8$. We have $B(v_1,0.8) = \{v_1,v_2,v_3\}$. Its boundary edges and vertices are $\xcut(v_1,0.8) = \{(v_2,v_4),(v_1,v_4)\}$ and $\ycut(v_1,0.8) = \{v_3\}$, respectively. Towards $\xvol(v_1,0.8)$, edge $(v_1,v_2)$ contributes $0.1$, $(v_1,v_3)$ contributes $0.3$, $(v_1,v_4)$ contributes $0.7$, and $(v_2,v_4)$ contributes $0.1$. Towards $\yvol(v_1,0.8)$, $v_1$ contributes $0.1$, $v_2$ contributes $0.5$, and $v_3$ contributes $0.4$.}
  \label{fig:region}
\end{figure}

The following lemma shows that there exists a \emph{good} radius $r \in [0,\rmax]$ such that the number of edges cut by $B(s,r)$ can be charged to $\xvol(s,r)$ and simultaneously, the number of vertices cut by $B(s,r)$ can be charged to $\yvol(s,r)$. \begin{lemma}[Region Growing Lemma]
  \label{lem:region}
  For every vertex $s$, we can efficiently find a \emph{good} radius $r \in [0,\rmax]$ such that
          \begin{align*}
    |\xcut(s,r)| &\leq O(\log \log n) \cdot \ln \frac{\vol_x(s,\rmax)}{\vol_x(s,r)} \cdot \vol_x(s,r),\\
    |\ycut(s,r)| &\leq \const \cdot \log w \cdot \vol_y(s,r),
  \end{align*}
  for some constant $\const$.
\end{lemma}

\begin{proof}
  Say that $r$ is \emph{$x$-good} if it satisfies the first inequality and \emph{$y$-good} if it satisfies the second inequality. Set $r(n):=\ln \ln [e (n + 1)]$ and define the following sets of radii:
  \begin{align*}
    A_x = &\Bigl\{r \in [0,\rmax] : \\ &\frac{|\xcut(s,r)|}{\xvol(s,r)} >
      \markov \cdot \ln \frac{e \cdot \vol_x(s,\rmax)}{\vol_x(s,r)}
      \cdot r(n) \Bigr\},
  \end{align*}
  \begin{align*}
A_y &= \left\{r \in [0,\rmax] : \frac{|\ycut(s,r)|}{\yvol(s,r)} > \markov \cdot \ln (w^2+1)\right\}.
  \end{align*}
In other words, $A_x$ and $A_y$ are the sets of radii that are $x$-bad and $y$-bad, respectively. We claim that the measure of both $A_x$ and $A_y$ are small: $\mu(A_x), \mu(A_y) \leq 1/\markov$. The claim then implies that there exists $r \in [0, \rmax]$ that is simultaneously $x$-good and $y$-good.

Observe that $\frac{\partial \xvol(s,r)}{\partial r} = |\xcut(s,r)|$ and $\frac{\partial \yvol(s,r)}{\partial r} = |\ycut(s,r)|$. Suppose, towards a contradiction, that $\mu(A_x) > 1/\markov$. We have
\begin{align*}
&\frac{\partial}{\partial r} \left(-\ln \ln \frac{e \cdot \xvol(s,\rmax)}{\xvol(s,r)}\right)\\
&= \frac{|\xcut(s,r)|}{\xvol(s,r)\cdot \ln \frac{e \cdot \xvol(s,\rmax)}{\xvol(s,r)}}
\end{align*}
So
  \begin{align*}
    &\ln \ln \frac{e \cdot \xvol(s,\rmax)}{\xvol(s,0)}\\
    &= -\ln \ln \frac{e \cdot \xvol(s,\rmax)}{\xvol(s,r)} \ \Bigg|^{\rmax}_0 \\
    &= \int_0^{\rmax} \frac{|\xcut(s,r)|}{\xvol(s,r)\cdot \ln \frac{e \cdot \xvol(s,\rmax)}{\xvol(s,r)}}\, \partial r \\
    &\geq \mu(A_x) \cdot \markov \ln \ln [e (n + 1)]
    >  r(n).
  \end{align*}
  But
  \[\ln \ln \frac{e \cdot \xvol(s,\rmax)}{\xvol(s,0)}
  \leq \ln \ln \frac{e (\xvol + \xvol/n)}{\xvol/n}
  = r(n).\]
  Thus, we have a contradiction and so $\mu(A_x) \leq 1/\markov$.

  We now turn to bounding $\mu(A_y)$. Suppose, towards a contradiction, that $\mu(A_y) > 1/\markov$.
  \begin{align*}
    \ln \frac{\yvol(s,1/\markov)}{\yvol(s,0)}
    &= \int_0^{\rmax} \frac{1}{\yvol(s,r)} \frac{\partial \yvol(s,r)}{\partial r}\, \partial r\\
    &= \int_0^{\rmax} \frac{|\ycut(s,r)|}{\yvol(s,r)}\, \partial r \\
    &\geq \mu(A_y) \cdot \markov \ln (w^2+1)
    > \ln (w^2+1).
  \end{align*}
  But \[\ln \frac{\yvol(s,1/\markov)}{\yvol(s,0)}
  \leq \ln \frac{w + 1/w}{1/w}
  = \ln(w^2 + 1).\]
  Therefore, both $\mu(A_x), \mu(A_y) \leq 1/\markov$ and so there exists $r \in [0, \rmax]$ that is simultaneously $x$-good and $y$-good.

  To find such an $r$ efficiently, note that as $r$ grows, $\vol_x(s,r)$ and $\vol_y(s,r)$ are non-decreasing. Thus, we only need to check the condition at points $r$ when either $\xcut(s,r)$ and $\ycut(s,r)$ changes, which happens at most $2|V'|$ and $2|E'|$ times, respectively. This is because as $r$ grows, once a vertex leaves $\ycut(s,r)$, it is inside $B(s,r)$ and will never reappear in $\ycut(s,r)$; and once an edge leaves $\xcut(s,r)$, both of its endpoints are inside $B(s,r)$ and the edge will never reappear in $\ycut(s,r)$.
\end{proof}

The next lemma shows that as we grow a ball $B(s,r)$ around vertex $s$, as long as $r \leq \rmax$, it cannot contain more than two-thirds of $S$.
\begin{lemma}
  \label{lem:radius}
  Let $U$ be a set of vertices. If $\max_{u,v \in U}d^S_{uv} \leq 1/6$, then $|U \cap S| \leq 2|S|/3$.
  \end{lemma}

\begin{proof}
  We prove the contrapositive. Suppose, that $|U \cap S| > 2|S|/3$. Let $u$ be a vertex in $U \cap S$. The spreading constraints imply that
  \[\sum_{v \in U \cap S} d^S_{uv} \geq |U \cap S| - \frac{|S|}{2} > |U \cap S| - \frac{3|U \cap S|}{4} = \frac{|U \cap S|}{4}.\]
Thus, by averaging, there exists $v \in U \cap S$ such that $d^S_{uv} > 1/4 > 1/6$.  \end{proof}

We are now ready to describe the $\algpart$ algorithm, which executes the Delete and Separate component of our algorithm. Recall that in the Delete and Separate step, the input is a subgraph $H = (V',E')$ and a vertex set $S \subseteq V'$, and our goal is find a set $D$ of edges to delete and a vertex set $X$ of size $O(w \log w)$ such that $X$ is a $\near$-separator of $S$ in $H-D$.

\paragraph{The $\algpart$ algorithm.}
\begin{enumerate}
\item Initialization. Let $\Hhat = H$
\item Region growing. While $\Hhat$ contains more than two-thirds of $S$,
  \begin{enumerate}
  \item Find region. Choose an arbitrary vertex $v \in S$ that is contained in $\Hhat$ and find a good radius $r$ such that $B(v,r)$ satisfies the conditions of Lemma \ref{lem:region}. Note that distances, balls and boundaries are defined with respect to $\Hhat$.
  \item Removal. Remove all vertices in $B(v,r)$ and their incident edges from $\Hhat$.
  \item Add zombies. For each edge $(s,u)$ that was removed, if $s \in \ycut(v,r)$ and $u$ is still in $\Hhat$, add \emph{zombie vertex} $z_u(s)$ with weight $y^S_{z_u(s)} = 0$ to $\Hhat$ and \emph{zombie edge} $(z_u(s), u)$ with length $x_{z_u(s),u} = x_{su}$.
  \end{enumerate}
\item Let $B(v_i,r_i)$ be the ball found in the $i$-th iteration.
\item Return the vertex set $X = \bigcup_i \ycut(v_i,r_i)$ and the edge set $D = \bigcup_i \xcut(v_i,r_i)$, replacing each zombie vertex $z_u(s)$ in $X$ with its original vertex $s$ and each zombie edge $(z_u(s), u)$ in $D$ with its original edge $(s,u)$.
\end{enumerate}

\begin{lemma}
  \label{lem:partitioning}
                  Suppose $\algpart$ took $\ell$ iterations and let $B(v_1,r_1), \ldots, B(v_\ell,r_\ell)$ be the regions it found. Let $C_1, \ldots, C_k$ be the components of $H-X-D$. For every $j \in [k]$, define the subgraph $H_j = (V_j,E_j)$ where $E_j$ is the subset of $E'-D$ with at least one endpoint in $C_j$, and $V_j$ is the set of endpoints of $E_j$.   We have the following:   \begin{enumerate}
  \item $|D| \leq O(\log \log n) \cdot\sum_{i=1}^\ell \ln\left(\frac{\xvol(H)}{\xvol(v_i,r_i)}\right) \xvol(v_i,r_i)$,
  \item $|X| \leq \Xconst \cdot (w +|S|/w) \log w$, for some constant $\Xconst$,
  \item $|C_i \cap S| \leq 2|S|/3$, and
  \item The edge set $E_j$ of each subgraph $H_j$ is either contained in a region $B(v_i,r_i)$ (and so $\xvol(H_j) \leq \xvol(v_i,r_i)$) or it is contained in the residual graph $\Hhat$ at the end of the execution of $\algpart$.
  \end{enumerate}
  \end{lemma}

\begin{proof}
The first statement follows from the fact that we chose a good radius $r_i$ for each region $B(v_i,r_i)$ and that the $x$-volume of any region can be at most $\xvol(H)$. Let us now consider the second statement. The fact that we chose a good radius for each region implies that
\[|X| = \sum_{i=1}^\ell |\ycut(v_i,r_i)| \leq \const \cdot \log w \sum_{i=1}^\ell \yvol(v_i,r_i)\]
for some constant $\const$.
Recall that $\yvol(v_i,r_i) = \yval(v_i,r_i) + 1/w$. Since each vertex can only contribute towards the weight of at most one region and zombie vertices have zero weight, the total weight of any region is at most $\sum_{v \in V'}y^S_v \leq w$. So, $\sum_{i=1}^\ell \yvol(v_i,r_i) \leq (w + \ell/w)$. Moreover, since the center of each region is a vertex $v_i$ of $S$, we remove at least one vertex of $S$ in each iteration. Thus, the number of iterations $\ell$ can be at most $|S|$. So, $|X| \leq \const \cdot \log w \cdot (w + \ell/w) \leq \const \cdot (w + |S|/w) \log w$ and this proves the second statement of the lemma.

Next, we argue that $|C_i \cap S| \leq 2|S|/3$. If $C_i$ was a component that remained at the end of the execution of $\algpart$, then by definition, $|C_i \cap S| \leq 2|S|/3$. Suppose $C_i$ was a component that is contained in $B(v_i,r_i)$. By Lemma \ref{lem:radius}, it suffices to check that the distance between any two remaining vertices $u$ and $v$ at the start of iteration $i$ is at least $d^S_{uv}$. Over the previous iterations, $\algpart$ modifies the graph in two ways: by removing edges and vertices, and by introducing zombie edges. Removing edges and vertices clearly cannot decrease the distance between $u$ and $v$. Zombie edges do not create a shortcut between $u$ and $v$ either. Thus, the distance between $u$ and $v$ in iteration $i$ is at least $d^S_{uv}$.

Finally, each subgraph $H_j$ is connected (since they are the set of edges with one endpoint in a connected component $C_j$), thus for each region $B(v_i,r_i)$, either all the vertices of $H_j$ are in $B(v_i,r_i)$ or all of them are not in $B(v_i,r_i)$.
\end{proof}

 \subsection{Putting it together}

We are now ready to describe the $\interdict$ algorithm, which recursively rounds a feasible solution to LP \eqref{lp:exp}. It takes as input a subgraph $H = (V',E')$ and a vertex set $S$, and deletes a set of edges $F$ such that $\tw(H-F) \leq O(w \log w)$. Let $\Xconst$ be the constant in Lemma \ref{lem:partitioning}.
\paragraph{The $\interdict$ Algorithm.}
\begin{enumerate}
  \item Delete and Separate.  Use algorithm $\algpart$ to find a set $D$ of edges to delete and a $\near$-separator $X$ of $S$ in $H-D$. Let $C_1, \ldots, C_k$ be the components of $H-D-X$. Delete $D$.
\item Define subproblems. For $i \in [k]$, define the subgraph $H_i = (V_i,E_i)$ where $E_i$ is the subset of $E'-D$ with an endpoint in $C_i$, and $V_i$ is the set of endpoints of $E_i$. \item Recurse. For $i \in [k]$, call $\interdict(H_i, V_i \cap (X \cup S))$. \end{enumerate}

To round a feasible solution of LP \eqref{lp:exp}, we call $\interdict(G,S_0)$ where $S_0$ is an arbitrary set of at most $O(w \log w)$ vertices. Let $F$ be the set of edges deleted by $\interdict(G,S_0)$.

\begin{lemma}
  \label{lem:tree-dec}
      $\tw(G - F) \leq O(w \log w)$.
\end{lemma}

\begin{proof}
    We can apply exactly the same analysis of the classic tree decomposition algorithm \cite{RobertsonS95b} to argue that the $O(w \log w)$-size separators found in the recursion can be used to construct a tree decomposition $T$ of $G-F$ such that the width of $T$ is at most $O(w \log w)$. \end{proof}

Recall that $\xval(G) = \sum_{(u,v) \in E} x_{uv}$, the cost of the LP solution $(\xvec, \yvec, \dvec)$. \begin{lemma}
  \label{lem:F-cost}
  $|F| \leq O(\log n\log \log n)\xval(G)$.
\end{lemma}

\begin{proof}
  Let $k$ be the recursion depth of $\interdict(G,S_0)$. For each depth $j$, let $\Hcal_j$ be the collection of subgraphs that $\interdict$ recursed on at depth $j$, and for each $H \in \Hcal_j$, let $R(H)$ be the set of regions found by $\algpart$ when $\interdict$ recursed on $H$. (Note that at depth $0$, we have $\Hcal_0 = \{G\}$.) By Lemma \ref{lem:partitioning}, we have the following bound on $|F|$:
  \begin{equation}
      |F| \leq O(\log \log n) \cdot \sum_j \sum_{H \in \Hcal_j} \sum_{B(s,r)
        \in
        R(H)}\ln\left(\frac{\xvol(H)}{\xvol(s,r)}\right)\xvol(s,r).\label{eq:F-cost}
\end{equation}
  For each edge $(u,v) \in E$, define $g(u,v) = x_{uv} + \frac{\xval(G)}{n^2}$. We have that
  \begin{equation}
  \sum_{(u,v) \in E}g(u,v) \leq 2\xval(G).\label{eq:total-g}
\end{equation}

We now show that $\sum_j \sum_{H \in \Hcal_j} \sum_{B(s,r) \in R(H)}\ln\left(\frac{\xvol(H)}{\xvol(s,r)}\right)\xvol(s,r) \leq O(\log n) \sum_{(u,v) \in E} g(u,v)$. We will do this by charging $\sum_{H \in \Hcal_j} \sum_{B(s,r) \in R(H)}\ln\left(\frac{\xvol(H)}{\xvol(s,r)}\right)\xvol(s,r)$ to edges of the subgraphs in $\Hcal_j$ and proving that every edge $(u,v)$ receives a total charge, across all depths, of at most $O(\log n) g(u,v)$. Our charging scheme is as follows: for each subgraph $H \in \Hcal_j$ and each region $B(s,r) \in R(H)$, charge $\ln\left(\frac{\xvol(H)}{\xvol(s,r)}\right) g(u,v)$ to each edge $(u,v)$ with at least one endpoint in $B(s,r)$.

Let us first see why the total amount charged per region $B(s,r) \in R(H)$ is at least $\ln\left(\frac{\xvol(H)}{\xvol(s,r)}\right)\xvol(s,r)$. Let $E_H(s,r)$ be the edges with at least one endpoint in $B(s,r)$. This is exactly the set of edges that contribute to the region's $x$-volume $\xvol(s,r)$. Therefore, 
\begin{align*}
  \sum_{(u,v) \in E_H(s,r)} g(u,v) 
  &\geq \sum_{(u,v) \in E_H(s,r)} x_{uv} + \frac{\xval(G)}{n^2}\\
  &\geq \xvol(s,r)
\end{align*}
and so the total amount charged is at least  $\ln\left(\frac{\xvol(H)}{\xvol(s,r)}\right)\xvol(s,r)$. Overall, we have
\begin{equation}
   \sum_j \sum_{H \in \Hcal_j} \sum_{B(s,r) \in R(H)}\ln\left(\frac{\xvol(H)}{\xvol(s,r)}\right)\xvol(s,r)
    \leq \sum_{(u,v) \in E} \phi(u,v),\label{eq:charge-uv}
 \end{equation}
   where $\phi(u,v)$ is the total charge received by $(u,v)$.

  Fix an edge $(u,v) \in E$. Let us now upper bound $\phi(u,v)$. Since $\interdict$ recurses on edge-disjoint subgraphs, $(u,v)$ is contained in at most one subgraph of $\Hcal_j$ for each depth $j$, and if it does not belong to a subgraph of $\Hcal_j$, then it does not belong to any subgraph of $\Hcal_{j'}$ at lower depths $j' > j$. In the worst case, $(u,v)$ is contained in some subgraph $H_j \in \Hcal_j$ for every depth $j$. Let us assume this is so and define $B(s_j,r_j) \in R(H_j)$ to be the region containing $(u,v)$. Lemma \ref{lem:partitioning} tells us that all of $H_j$ is contained in the region $B(s_{j-1},r_{j-1})$ for each $j$ and so $\xvol(s_j, r_j) \geq \xvol(H_{j+1})$. Thus,
  \begin{align*}
    \phi(u,v)
    &= \sum_{j=0}^k \ln\left(\frac{\xvol(H_j)}{\xvol(s_j,r_j)}\right) g(u,v)\\
    &\leq \sum_{j=0}^k \ln\left(\frac{\xvol(H_j)}{\xvol(H_{j+1})}\right) g(u,v)\\
    &\leq \ln\left(\frac{\xvol(G)}{\xvol(H_{k+1})}\right) g(u,v) \\
    &\leq \ln (n^2+1) g(u,v),
  \end{align*}
  where the last inequality follows from the fact that $x$-volume is always at least $\frac{\xval(G)}{n^2}$.

Combining this with Inequality \eqref{eq:charge-uv}, we get
  \begin{align*}
    &\sum_j \sum_{H \in \Hcal_j} \sum_{B(s,r) \in R(H)}\ln\left(\frac{\xvol(H)}{\xvol(s,r)}\right)\xvol(s,r)\\
    &\leq \sum_{(u,v) \in E} \ln (n^2+1) g(u,v)\\
    &\leq 2\ln (n^2+1) \xval(G),
  \end{align*}
  where the last inequality follows from Inequality \eqref{eq:total-g}. Finally, plugging this into the right hand side of \eqref{eq:F-cost} gives us $|F| \leq O(\log n \log \log n)\xval(G)$, as desired.
\end{proof}

Lemmas \ref{lem:F-cost} and \ref{lem:tree-dec} imply Lemma \ref{lem:tw-oracle}.

 \section{Using the LP}
\label{sec:LP}
We now come to the problem of how to handle the LP \eqref{lp:exp} in $n^{O(1)}$ time.
As discussed in the Introduction, we bypass the need to completely solve this LP using the
Round-or-Separate framework.
A crucial ingredient here is that if the rounding step gets stuck (is unable to find a small separator), we need to find a violated inequality for the $x$-variables, and not just for the LP \eqref{lp:exp}. In Section \ref{s:lp:part1} we show how to reformulate LP  \eqref{lp:exp} to only have the $x$ variables.
In this reformulation, the coefficients of the inequalities come from feasible points in a polytope with exponentially many variables, so Section \ref{s:lp:part2} deals with how to fix such a violated inequality.

\subsection{Reformulating the LP}
\label{s:lp:part1}
 We first reformulate LP \eqref{lp:exp}---the original LP with variables $(\xvec,\yvec,\dvec)$---to obtain another LP with only $x$-variables.

Call a vector $\xvec$ \emph{feasible} if there exist vectors $\yvec$ and $\dvec$ such that $(\xvec, \yvec, \dvec)$ is a feasible solution to LP \eqref{lp:exp}.
Define $\feas$ to be the set of all feasible vectors $\xvec$, and observe that $\feas$ is simply the feasible region of LP \eqref{lp:exp} with the $y$ and $d$ variables projected out.

Next, we show how to describe $\feas$ using linear inequalities. For every vertex set $S$,  we define an LP parameterized by $\xvec$ (see Figure \ref{fig:lpS}). We call this $\sepLP(\xvec, S)$.
\begin{figure*}[htp]
  \centering
\begin{equation}
    \boxed{
      \label{lp:S}
      \begin{aligned}
        \mbox{min} \quad
        & \sum_v y_v\\
        \mbox{s.t.}\quad
        & d_{uv} \leq \sum_{e \in E(P)} x_e + \sum_{t \in V(P)} y_t
        &\quad \forall u,v \in V, P \in \paths(u,v) \\
        & \sum_{u \in T} d_{uv} \geq |T| - |S|/2
        & \quad \forall T \subseteq S \subseteq V, v \in T
      \end{aligned}
     }
  \end{equation}
  \caption{$\sepLP(\xvec,S)$}
  \label{fig:lpS}
\end{figure*}
We emphasize that in this LP, only $\yvec$ and $\dvec$ are variables.
Recall that this LP is related to the problem of finding a small balanced vertex separator of $S$ in $G$,
provided the edges are removed fractionally to extent $x_e$.

\begin{definition}
Let $\sepOPT(\xvec,S)$ be the value of the optimum solution to the $\sepLP(\xvec, S)$.
We say that $\xvec$ is \emph{$S$-feasible} if $\sepOPT(\xvec, S) \leq w$, and denote by $\feas(S)$ the set of $S$-feasible vectors.
\end{definition}
\begin{lemma}
  \label{lem:equiv}
  $\feas = \bigcap_{S \subseteq V} \feas(S)$.
\end{lemma}

\begin{proof}
  Suppose $\xvec \in \feas$. Then, for every vertex set $S$, the vectors $\yvec^S$ and $\dvec^S$ are a feasible solution to $\sepLP(\xvec,S)$ and $\sum_v y^S \leq w$, so $\sepOPT(\xvec,S) \leq w$.

  On the other hand, suppose $\sepOPT(\xvec, S) \leq w$ for all vertex sets $S$. Define $\yvec$ and $\dvec$ such that for each $S \in \Scal$, $(\yvec^S, \dvec^S)$ is the optimal solution to $\sepLP(\xvec,S)$. As $\sepOPT(\xvec, S) \leq w$ for all $S \in \Scal$, we have that $(\xvec, \yvec, \dvec)$ is a feasible solution to LP \eqref{lp:exp}.
\end{proof}

Thus, to describe $\feas$ using linear inequalities, it suffices to describe $\feas(S)$ using linear inequalities. By linear programming duality, we have that $\sepOPT(\xvec,S) \leq w$ if and only if for every feasible solution to the dual of $\sepLP(\xvec,S)$ has objective value at most $w$.

So we consider the dual LP \eqref{lp:S-dual} (Figure \ref{fig:flowLP}), and denote it by $\flowLP(\xvec,S)$. (We call this $\flowLP$ as it is dual to a ``cut-type'' LP.) 

\begin{lemma} \label{lem:dual} The dual to $\sepLP(\xvec,S)$ is given by LP \eqref{lp:S-dual} below.
\end{lemma}

\begin{proof}
Let us introduce the variables $f^{uv}_P$ for the first set of constraints in LP \eqref{lp:S}, and the variables 
$ g_{T,v}$ for the second set of constraints.
Again, we emphasize that only $f$ and $g$ are variables here (and that $x$ is not a variable).

Let us check that the dual is exactly LP \eqref{lp:S-dual}. Rewrite the first set of primal constraints as 
\[\sum_{t \in V(P)} y_t - d_{uv} \geq -\sum_{e \in E(P)} x_e \quad \forall u,v \in V, P \in \paths(u,v)\]  
In the objective, the coefficient of $f^{uv}_P$ is $-\sum_{e \in E(P)} x_e$ and the coefficient of $g_{T,v}$ is $|T| - |S|/2$. The dual constraints correspond to 
primal variables $d_{uv}$ and $y_t$. 

In the primal, the variable $d_{uv}$ has a coefficient of $-1$ for the constraints corresponding to the dual variables $f^{uv}_P$ for $P \in \paths(u,v)$, and a coefficient of $1$ in the constraints corresponding to the dual variables $g_{T,t}$ for $t \in \{u,v\}$ and $T \ni t$. Similarly, the variable $y_t$ has a coefficient of $1$ for the constraints corresponding to the dual variables $f^{uv}_P$ for $u,v \in V$ and $P \in \paths(u,v)$ such that $t \in V(P)$. So, LP \eqref{lp:S-dual} is dual to LP \eqref{lp:S}. 
\end{proof}
\begin{figure*}
  \centering
\begin{equation}
    \boxed{
      \label{lp:S-dual}
      \begin{aligned}
        \mbox{max} \quad
        &\sum_{T \subseteq S, v \in T} g_{T,v}(|T| - |S|/2) 
        - \sum_{u,v \in V; P \in \paths(u,v)} f^{uv}_P\left[\sum_{e \in P} x_e\right]\\
        \mbox{s.t.}\quad
        & \sum_{T \subseteq V : u, v \in T} g_{T,u} + g_{T,v}
        \leq \sum_{P \in \paths(u,v)} f^{uv}_P
        &\quad \forall u,v \in V,\\
        & \sum_{u,v \in V} \sum_{P \in \paths(u,v): t \in V(P)} f^{uv}_P \leq 1
        &\quad \forall t \in V\\
      \end{aligned}
    }
  \end{equation}
  \caption{$\flowLP(\xvec,S)$.}
  \label{fig:flowLP}
\end{figure*}

\paragraph{Separation oracle for $\feas(S)$.}  Given a solution $(\fvec, \gvec)$ to $\flowLP(\xvec,S)$, we denote its objective value as $\flowLP(\xvec,S,\fvec,\gvec)$. LP duality implies the following lemma.
\begin{lemma}
  \label{lem:ineqs}
  $\xvec \in \feas(S)$ if and only if $\flowLP(\xvec,S,\fvec,\gvec) \leq w$ for all feasible dual solutions $(\fvec, \gvec)$.
\end{lemma}

Let us see what we have achieved so far. The expression $\flowLP(\xvec,S,\fvec,\gvec) \leq w$ is a linear constraint on $x$, whose coefficients  $(\fvec, \gvec)$ are feasible points in the polytope given by \eqref{lp:S-dual}. Now we crucially note that the coefficients of the constraints on $(\fvec, \gvec)$ depend only on the topology of the graph $G$ and not on $\xvec$.
\begin{definition}
 We say that $(\fvec, \gvec)$ is \emph{$S$-valid} if it satisfies the constraints of $\flowLP(\xvec,S)$.
 \end{definition}

Thus, Lemma \ref{lem:ineqs} gives a description of $\feas(S)$ in terms of linear inequalities. Together with Lemma \ref{lem:equiv}, we get a description of $\feas$ in terms of linear inequalities and so we can reformulate LP \eqref{lp:exp} as follows.
\begin{equation}
\boxed{
  \label{lp:poly}
\begin{aligned}
  \mbox{min} \quad
  & \sum_{(u,v) \in E} x_{uv}\\
  \mbox{s.t.}\quad
  & \flowLP(\xvec,S,\fvec,\gvec)\leq w  \quad\forall S \subseteq V, (\fvec, \gvec) \mbox{ $S$-valid}
\end{aligned}
}
\end{equation}

\subsection{Round or Separate}
\label{s:lp:part2}
We will work with the reformulated LP \eqref{lp:poly}. One immediate obstacle is that it is unclear how to get an efficient separation oracle for $\feas$.
In fact, even for a fixed vertex set $S$, it is not immediately clear how to find a violated inequality for $\feas(S)$. The problem is that for any constraint, the coefficients of the $x$-variables are based on $\fvec,\gvec$, which need to satisfy LP \eqref{lp:S-dual}, and it is not clear how to generate them; e.g. there are exponentially many $\fvec, \gvec$ variables in $\flowLP(\xvec,S)$. To get over these problems, we do the following.

We apply the Round-or-Separate framework with the $\interdict$ algorithm in Section \ref{sec:ROS}. Roughly speaking, we start with some candidate solution $\xvec$ (possibly infeasible), and try to construct a tree decomposition of width $O(w \log w)$ using $\interdict$. If we succeed, we are done. Otherwise, $\interdict$ could not find a small balanced separator for some vertex set $S$, and this can only happen if $\xvec \notin \feas(S)$. Later, in Lemma \ref{lem:sep-oracle}, we give an efficient procedure that given $\xvec$ and $S$, determines whether $\xvec \in \feas(S)$ and if so, outputs a solution $(\yvec^S, \dvec^S)$ that is feasible to $\sepLP(\xvec,S)$ and satisfies $\sum_v y^S_v \leq w$; otherwise, it outputs a separating hyperplane, i.e. an $S$-valid $(\fvec,\gvec)$ such that $\flowLP(\xvec,S,\fvec,\gvec) > w$. In the first case, $\interdict$ can use $(\yvec^S, \dvec^S)$ to find a $O(w \log w)$-size separator of $S$ and make progress. In the second case, we can add the separating hyperplane to find a new candidate $x$ and repeat the whole tree decomposition procedure. By the standard Ellipsoid method-based, separation-versus-optimization framework, the number of such iterations is polynomially bounded.

It remains to show how to efficiently separate $\feas(S)$ for any given subset $S$.
We will do this indirectly by using the Ellipsoid Method and applying duality to $\sepLP(\xvec, S)$.

\begin{lemma}
  \label{lem:sep-oracle}
  There exists a polynomial-time algorithm that given $\xvec$ and $S$, determines whether $\xvec \in \feas(S)$, and if so, outputs a solution $(\yvec^*, \dvec^*)$ that is feasible to $\sepLP(\xvec,S)$ and satisfies $\sum_v y^*_v \leq w$; otherwise, it finds a violated inequality for LP \eqref{lp:poly}. \end{lemma}

\begin{proof}
  We apply the Ellipsoid method to find an optimal solution $(\yvec^*,\dvec^*)$ to $\sepLP(\xvec, S)$. We can do this efficiently by using a separation oracle that uses Dijkstra's algorithm to determine the distances $\dvec$ given $\xvec$ and $\yvec$. If $\sum_v y^*_v \leq w$, then $\xvec \in \feas(S)$ by definition. 
  
  On the other hand, if $\sum_v y^*_v > w +\epsilon$ (where $\epsilon$ can be made exponentially small), then we can find a violated inequality for \eqref{lp:poly} as follows. As
 the Ellipsoid method added only polynomially many constraints we can re-solve this LP only on these added constraints and assume that $\yvec^*$ is a solution to this smaller LP on only polynomially many constraints.
     By complementary slackness, there exists an optimal dual solution $(\fvec^*, \gvec^*)$ where the only non-zero dual variables are those that correspond to these polynomially many primal constraints. Thus, it suffices to solve $\flowLP(\xvec,S)$ restricted to these dual variables which makes $\flowLP(\xvec,S)$ polynomial in size. The objective function in $\flowLP$ then gives the violated inequality for LP \eqref{lp:poly}.
\end{proof}

\section{Bounded Size Interdiction and MIS in Noisy Planar Graphs}
\label{sec:MIS}
We consider the Bounded Size Interdiction problem and show how it implies Theorem \ref{thm:MIS}.

Consider the noisy graph $G$ obtained by adding $\delta n$ edges to some planar $G_0$. 
Let us view $G$ as obtained by superimposing the noisy edges on the recursive decomposition of $G_0$ given by Lemma \ref{thm:LT}. This directly implies the following (noisy) decomposition for $G$.
\begin{lemma}[Noisy Decomposition] 
\label{lemma:rough-decompo} Given a $\delta$-noisy planar graph $G$, 
for any $\alpha >0$, there exists a partition $X,C_1,\ldots,C_k$ of $V$ with (i) $|X| \leq c\alpha n$ for some universal constant $c$, (ii)
 $|C_i| \leq 1/\alpha^2$ for all $i\in [k]$, and (iii) at most $\delta n$ edges whose endpoints lie in distinct (two different) $C_i$s.
\end{lemma}
 
Of course as we do not know $G_0$, it is not clear how to find such a decomposition. Theorem \ref{lem:noisy-LT} shows that this can be done approximately.

\begin{theorem}[Noisy Bounded Size Interdiction]
  \label{lem:noisy-LT}
Let $G = (V,E)$ be any graph that has a vertex partition $X,C_1,\ldots,C_k$ with $|C_i| \leq s$ for each $i \in [k]$, and let $b$ be the total number of edges 
whose endpoints lie in distinct $C_i$s.

Then  for every $\beta \leq 1$, we can find in time $n^{O(s)}$ a vertex partition $X',C'_1,\ldots,C'_{k'}$ such that (i) $|X'| = O(|X|/\beta)$, (ii)
 $|C'_i| \leq s$ for each $i\in [k']$, and (iii) at most $O(b + \beta |E|)$ edges whose endpoints lie in distinct $C_i$s. 
\end{theorem}

This implies the following proper decomposition (where the $C'_i$ are components of $G[V-X]$).
\begin{corollary}
\label{cor:bis} There is a subset $X' \subset V$ of size $O(|X|/\beta + b + \beta |E|)$ such that the components in $G[V-X']$ have size at most $s$.
\end{corollary}
\begin{proof}
For each edge $(u,v)$ with $u \in C'_i$ and $v \in C'_j$ for $i\neq j\in [k']$, 
put an endpoint (say $u$) in $X'$ and remove $u$ from $C'_i$. 
As a result, $X'$ has size  $ O(|X|/\beta + b + \beta |E|)$ and there are no edges between different $C'_i$ and $C'_j$.
\end{proof}

Before we prove Theorem \ref{lem:noisy-LT}, let us first see how it implies Theorem \ref{thm:MIS}.

\smallskip

\begin{proofof}[Theorem \ref{thm:MIS}]
Given an $\epsilon > 0$, set $\gamma = O(\epsilon^2)$.
Applying Lemma \ref{lemma:rough-decompo} with $\alpha=\gamma$, $G$ has a noisy decomposition with $|X| = O(\gamma n)$ and pieces $C_i$ of size $s=1/\gamma^2$ and at most $b=\delta n$  edges between these $C_i$'s. Moreover $G$ has at most $(3 + \delta )n \leq 4n$ edges.
Applying Corollary \ref{cor:bis} with $\beta = \gamma^{1/2}$, gives an $X'$ of size $O(|X|/\gamma^{1/2} + \delta n + 
\gamma^{1/2} n) = O( (\epsilon + \delta)n)$,
such that the components of $G[V-X']$ have at most $1/\gamma^2 = 1/\epsilon^4$ vertices.

By exhaustive search, find an MIS in each $C_i$ separately and return their union $I$. Clearly, $I$ is a valid independent set. By Theorem \ref{lem:noisy-LT}, the algorithm has overall running time $n^{O(1/\gamma^2)} = n^{O(1/\epsilon^4)}$.
As $G_0$ is planar and hence $4$-colorable, and $G\setminus G_0$ has at most $\delta n$ noisy edges, $\alpha(G) \geq \alpha(G_0) - \delta n \geq n/4 - \delta n.$
Moreover, as $|I| \geq \alpha(G) - |X'|$, this gives $|I| \geq (1- O(\epsilon +\delta)) \alpha(G)$.      
\end{proofof}

\noindent
{\em Remark.} Theorem \ref{thm:MIS} extends directly to the minor-free case using the separator theorem for minor-free graphs \cite{AST}, and the fact that these graphs have bounded average degree and thus $\alpha(G) = \Omega(n)$.

 \subsection{Proof of Theorem \ref{lem:noisy-LT}}
\label{sec:noisy-LT}

\paragraph{LP formulation.}	
Given $0 < \beta \leq 1$ and $G$ as input, we first write an integer program to find $X$ and the $C_i$s.
Let $a=|X|$ (we can assume that $a$ is known as the algorithm can try every value).
For each vertex $v$, the variable $y_v$ indicates if $v \in X$.
For each subset $S\subset V$ with $|S|\leq s$, the variable $z_S$ indicates if $S$ is one of the pieces $C_i$. Let $\mathcal{S}$ be the collection of all such subsets $S$.
For each $(u,v)\in E$, the variable $x_{uv}$ indicates whether the edge $(u,v)$ is such that $u \in C_i$ and $v \in C_j$ for some $i \neq j$.

Consider the integer program IP \eqref{lp:partition} in Figure \ref{fig:lppart}.
\begin{figure*}
  \centering
\begin{equation}
    \boxed{
      \label{lp:partition}
      \begin{aligned}
        \mbox{min} \quad
        & \sum_{(u,v) \in E} x_{uv}\\
        \mbox{s.t.}\quad
        &  \sum_v y_v \leq a \quad \\
        & \sum_{S : v \in S } z_S = 1 - y_v &\quad \forall   v \in V\\
        & \sum_{S : u \in S \wedge v \notin S} z_S \leq x_{uv} + y_v &\quad \forall (u,v) \in E\\
        & x_{uv},y_v,z_S \in \{0,1\} &\quad \forall (u,v) \in E, v \in
        V, S \in \mathcal{S}
      \end{aligned}
    }
  \end{equation}
  \caption{IP for the Bounded Size Interdiction Problem.}
  \label{fig:lppart}
\end{figure*}
                                    This is easily seen to be a valid formulation for the problem.
The first set of constraints ensure that $X$ contains at most $a$ vertices. The second set of constraints ensure that each vertex lies  in either $X$ or some $C_i$.
The third set of constraints are more involved and force $x_{uv}$ to be $1$ if some edge has endpoints in distinct $C_i$ and $C_j$. In particular, it says that if $u$ lies in some $C_i$ and $v$ does not lie in that $C_i$, then either $v$ lies in $X$ (i.e.~$y_v=1$) or $x_{uv}=1$. 
Note that the third set of constraints are asymmetric in $u$ and $v$, and we will put two such constraints (with $u$ and $v$ swapped) for each edge $(u,v)$.

As the objective function exactly measures the number of edges with endpoints in
two different $C_i$'s, it follows that the IP above has a feasible solution with value at most $\delta n$.

Consider the LP relaxation of this program. It has $O(n^{s})$ variables, and $O(n^2)$ non-trivial constraints. 
So in time $n^{O(s)}$, we can find some basic feasible solution with support size at most $O(n^2)$.
We reuse $x_{uv}$, $y_v$ and $z_S$ to denote some fixed optimum solution to the LP.

\paragraph{The Rounding Algorithm.}
The algorithm will construct the required $X'$ and the collection $\mathcal{C}$ of sets $C'_i$ from the LP solution by the following preprocessing and sampling procedure.

\begin{enumerate}
\item Initialization. We set $X',\mathcal{C} = \emptyset$ and $U =V$, where $U$ denotes the set of vertices not covered by $\mathcal{C}$.
 \item Preprocessing. Add every vertex $v$ with $y_v \geq \beta$ to $X'$. 
Set $U = U \setminus X'$
\item Sampling to create $\mathcal{C}$. 
Arbitrarily order the sets $S_1,\ldots,S_k$ in the support of the LP solution.
Repeat the following (phase) until $U$ is empty:

{\em Phase.} For $i=1,\ldots,k$, sample the set $S_i$ randomly with probability $z_{S_i}$. If $S_i$ is picked,
add $C'=S_i \cap U$ to the collection $\mathcal{C}$, and update  $U = U \setminus S_i$.
\end{enumerate}

\paragraph{Analysis.}
Clearly the sets $C'$ produced by the algorithm have size at most $s$, and they are disjoint.

\begin{lemma}
\label{small:x}
$|X'| \leq a/\beta$.
\end{lemma}
\begin{proof}
As $X'$ is the set of vertices $v$ with $y_v \geq \beta$, there can be at most $a/\beta$ such vertices by the LP constraint $\sum_v y_v \leq a$.
\end{proof}

Henceforth, we also assume that $\beta \leq 1/2$, otherwise choosing $X' = \emptyset$  and partitioning $V$ arbitrarily into sets $C_1,\ldots,C_k$ of size at most $s$  trivially suffices for Theorem \ref{lem:noisy-LT}.

We now show that the algorithm runs in expected polynomial time and does not generate a vertex partition with too many edges between distinct $C'_i$.
\begin{lemma}
\label{small:f}
Let $F$ be the set of edges with endpoints in two distinct $C'_i$s. Then  $\E[|F|] \leq O(b + \beta|E|) $.
Moreover, the algorithm terminates after at most $O(n^2 \log n)$ sampling steps with high probability.
\end{lemma}
\begin{proof}
We claim that $U$ is empty after $O(\log n)$ phases, with high probability.
After the preprocessing step, each uncovered vertex in $U$ has $y_v < \beta \leq 1/2$.
Thus, by the second LP constraint, $p_v:=\sum_{ S\ni v} z_S  = 1- y_v \geq 1/2$.
So the probability that a vertex $v$ is not covered after $j$ phases is  
\[ \left(\prod_{S \ni v} (1-z_S)\right)^j  \leq \exp\left(-j p_v \right) \leq \exp \left(-j/2\right) \leq (2/3)^j \] 
The claim now follows from a union bound over the $n$ vertices.

We now bound the size of $F$. 
    Let us focus on an edge $e=(u,v)$ and bound the probability that it is \emph{cut}, that is, added to $F$ during the Sampling step. Let $U_j$ denote the vertices in $U$ at the end of phase $j$.
The edge is cut in phase $j$ if and only if both $u$ and $v$ remain in $U$ at the end of phase $j-1$ (i.e. $u,v \in U_{j-1}$) and a set $S$ with $|S \cap \{u,v\}| = 1$ is chosen in phase $j$. 
As $\Pr[u,v \in U_{j-1}] \leq \Pr[v \in U_{j-1}] \leq (2/3)^{j-1}$, this implies that 
\begin{equation}
\label{eq:phasecut}
\Pr[(u,v) \mbox{ cut in phase $j$}] 
\leq (2/3)^{j-1}  \cdot \sum_{S : |S \cap \{u,v\}| = 1}z_S.
\end{equation}
Moreover, by the third set of constraints in LP \eqref{lp:partition}
  \begin{equation}
    \sum_{S : |S \cap \{u,v\}| = 1}z_S = \sum_{S : u \in S \wedge v
      \notin S}z_S + \sum_{S : v \in S \wedge u \notin S}z_S 
    \leq
    2x_{uv} + y_u + y_v.\label{eq:cut}
  \end{equation}
Summing \eqref{eq:phasecut} over all the phases and using \eqref{eq:cut}, we get 
\[ \Pr[(u,v) \mbox{ cut}] \leq 3(2x_{uv} + y_u + y_v)  \leq 6 x_{uv} + 6\beta, \]
where the second inequality follows as both $u$ and $v$ were not chosen in $X$ during the preprocessing step and hence $y_u,y_v \leq \beta$. 
By linearity of expectation, this implies that
\[ \E[|F|] = \sum_{(u,v) \in E} (6 x_{uv} + 6 \beta) = O(b + \beta|E|).\]
\end{proof}

\section*{Acknowledgements}
We would like to thank Marcin Pilipczuk and Fedor Fomin for many useful discussions and pointers, and anonymous reviewers for bringing recent local search-based techniques to our attention. Part of this work was done when all the authors were visiting the Simons Institute for the Theory of Computing.

\bibliographystyle{plain}
\bibliography{noisy_MIS}

\end{document}